\documentclass[a4paper,UKenglish,cleveref, autoref]{lipics-v2019}
\title{Reachability Games in Dynamic Epistemic Logic}

 \author{Bastien Maubert}{Universit\`a degli Studi di Napoli
   ``Federico II'', Italy}{}{}{}
 \author{Sophie Pinchinat}{Universit\'e de Rennes, France}{}{}{}
 \author{Fran\c{}cois Schwarzentruber}{ENS Rennes, France}{}{}{}

\nolinenumbers
 
\authorrunning{B. Maubert, S. Pinchinat and F. Schwarzentruber}%TODO mandatory. First: Use abbreviated first/middle names. Second (only in severe cases): Use first author plus 'et al.'

\Copyright{Bastien Maubert, Sophie Pinchinat and Fran\c{}cois Schwatzentruber}%TODO mandatory, please use full first names. LIPIcs license is "CC-BY";  http://creativecommons.org/licenses/by/3.0/

%TODO mandatory; please add comma-separated list of keywords

\category{}%optional, e.g. invited paper

\relatedversion{}%optional, e.g. full version hosted on arXiv, HAL, or other respository/website
\supplement{}%optional, e.g. related research data, source code, ... hosted on a repository like zenodo, figshare, GitHub, ...

\hideLIPIcs  %uncomment to remove references to LIPIcs series (logo, DOI, ...), e.g. when preparing a pre-final version to be uploaded to arXiv or another public repository

\usepackage{amsmath}
\usepackage{booktabs}
\usepackage{algorithm}
\usepackage{algorithmic}

\usepackage{ifthen}

\newif\ifdraft\drafttrue
\usepackage{manfnt}
\usepackage{draftmode}

\newboolean{longversion}
\setboolean{longversion}{true}
\newcommand\islongversion[1]{\ifthenelse{\boolean{longversion}}{#1}{}}
\newcommand\isnotlongversion[1]{\ifthenelse{\boolean{longversion}}{}{#1}}

 \newenvironment{proofsketch}{\trivlist\item[\hskip \labelsep {\bf
 Proof sketch.}\enskip]}%
 {\unskip\nobreak\hskip 2em plus 1fil\nobreak%
 \fbox{\rule{0ex}{0.3ex}\hspace{0.3ex}\rule{0ex}{0.3ex}}%
 \parfillskip=0pt \endtrivlist}

\usepackage{xspace} %space after a command
\usepackage{eurosym} %for the symbol euro
\usepackage{chessboard}
\usepackage{stmaryrd}

\usepackage{calc} % for doing computations in LaTEX :)
\usepackage{url}
\usepackage{color}

\usepackage{tikz}
\usetikzlibrary{arrows,patterns,topaths,automata}
\usepackage{graphicx}

\usepackage{array} % extends environments «array» et «tabular»

\usepackage{amsmath}
\usepackage{amssymb}  % \mathbb
\DeclareMathAlphabet{\mathpzc}{OT1}{pzc}{m}{it} % \mathpzc
\usepackage{mathrsfs} % \mathscr

\usepackage{fancyvrb}
\usepackage[normalem]{ulem}

\usepackage{colortbl}

\usepackage{multirow}
\usepackage{enumitem}

\newcommand{\egdef}{:=}

\newcommand{\agent}{{a}}
\newcommand{\agenta}{\agent}
\newcommand{\agentb}{{b}}
\newcommand{\atom}{p}

\newcommand{\atmset}{\ensuremath{\mathit{AP}}\xspace}
\newcommand{\atmsetf}{\ensuremath{\mathit{AP}_u}\xspace}
\newcommand{\AP}{\atmset}

\newcommand{\lang}{\mathcal{L}}
\newcommand{\languagePropositional}{\lang_{\text{Prop}}}

\newcommand{\languageEL}{\ensuremath{\lang_{\mathbf{EL}}}}

\newcommand{\kripkemodel}{\mathcal M}
\newcommand{\aworld}{w}

\newcommand{\world}{\aworld}

\newcommand{\worldb}{u}
\newcommand{\setworlds}{W}

\newcommand{\epistemicrelation}[1]{R_{#1}}
\newcommand{\epistemicrelationequiv}[1][\agent]{\sim_{#1}}
\newcommand{\relworlds}{\textit{R}}
\newcommand{\relworldsa}[1][\agent]{\relworlds_{#1}}
\newcommand{\valuationfunction}{V}

\newcommand{\eventmodel}{\mathcal A}
\newcommand{\eventmodeltuple}{(\setevents, (\epistemicrelationevents{\agent})_{\agent \in \agtset}, \pre, \post)}

\newcommand{\pre}{\text{pre}}
\newcommand{\post}{\text{post}}
\newcommand{\setevents}{A}
\newcommand{\epistemicrelationevents}[1]{{R^\eventmodel_{#1}}}
\newcommand{\epistemicrelationeventsequiv}[1][\agent]{\sim^{\eventmodel}_{#1}}

\newcommand{\eventlettreminuscule}{\alpha}
\newcommand{\event}{\eventlettreminuscule}
\newcommand{\eventb}{\eventlettreminuscule'}
\newcommand{\assignment}{{:=}}

\definecolor{TODO_COLOR}{rgb}{1,0.5,0.5}

\newcommand{\union}{\cup}

\newcommand{\set}[1]{{\{#1\}}}

\newcommand{\vide}{\emptyset}

\newcommand{\suchthat}{\mid}
\newcommand{\suth}{s.t.\ }

\newcommand{\bigand}{\bigwedge}
\newcommand{\lbigand}{\bigand}

\newcommand{\bigor}{\bigvee}

\newcommand{\bottom}{\bot}

\renewcommand{\phi}{\varphi}

\newcommand{\tuple}[1]{\langle #1 \rangle}

\newcommand{\lknow}[1]{K_{#1}}
\newcommand{\lknowpos}[1]{\hat K_{#1}}

\newcommand{\formula}{\phi}

\newcommand{\agtset}{\ensuremath{\textit{Agt}}\xspace}
\newcommand{\agtsetexists}{\ensuremath{\textit{Agt}}_\exists\xspace}
\newcommand{\agtsetforall}{\ensuremath{\textit{Agt}}_\forall\xspace}

\newcommand{\grammaris}{\hspace{2mm}::=\hspace{2mm}}
\newcommand{\grammarseparation}{\hspace{2mm}\mid\hspace{2mm}}
\newcommand{\grammarsep}{\grammarseparation}

\newcommand{\algofunction}{\textbf{function }}

\newcommand{\algoprocedure}{\textbf{procedure }}

\newcommand{\algoendfunction}{\textbf{endFunction }}

\newcommand{\algofor}{\textbf{for }}
\newcommand{\algodo}{\textbf{do }}

\definecolor{algocommentbackgroundcolor}{rgb}{1,1,0.5}

\newcommand{\algowhile}{\textbf{while }}

\newcommand{\algoif}{\textbf{if }}
\newcommand{\algothen}{\textbf{then }}
\newcommand{\algoelse}{\textbf{else }}

\newcommand{\algomatch}{\textbf{match }}

\newcommand{\algocase}{\textbf{case }}

\newcommand{\algoaccept}{\textbf{accept }}
\newcommand{\algoreject}{\textbf{reject }}

\newlength{\algoindentlongueur}
\newcommand{\algoindent}{\hspace*{\algoindentlongueur}}

\newlength{\algoindentavantvrulelongueur}
\newcommand{\algoindentavantvrule}{\hspace*{\algoindentavantvrulelongueur}}

\newlength{\dummy}

\newsavebox{\frameminipageboiteavecunnomsuperlongdesortequonnepuissepaslereutiliser}
\newenvironment{frameminipage}[2][c]{%
\begin{lrbox}{\frameminipageboiteavecunnomsuperlongdesortequonnepuissepaslereutiliser}%
\begin{minipage}[#1]{#2}%
} {%
\end{minipage}%
\end{lrbox}%
\framebox{\usebox{\frameminipageboiteavecunnomsuperlongdesortequonnepuissepaslereutiliser}}%
}

\newenvironment{algo} {
  \begin{frameminipage}{\linewidth}
} {
  \end{frameminipage}
}

\newenvironment{algobloc}{\setlength{\dummy}{\linewidth}\addtolength{\dummy}{- \algoindentlongueur}\addtolength{\dummy}{- \algoindentavantvrulelongueur}\algoindentavantvrule\vrule\algoindent\begin{minipage}{\dummy}}{\end{minipage}}

\newenvironment{algoblocprocedure}[1]
{\algoprocedure #1 \\  \begin{algobloc}}
	{\end{algobloc} }

\newenvironment{algoblocfor}[1]
{\algofor #1 \algodo \\  \begin{algobloc}}
{\end{algobloc} }

\newenvironment{algoblocwhile}[1]
{\algowhile #1 \algodo \\  \begin{algobloc}}
	{\end{algobloc} }

\newcommand{\indexemph}[1]{\emph{#1}\index{#1}}

\newcommand{\valuationelement}{v}

\tikzstyle{world}=[inner sep=0.5mm]
\tikzstyle{event}=[fill=gray!30, inner sep=0.5mm]
\tikzstyle{realworldarrowfromleft} = [initial left, initial text={}]
\tikzstyle{realworldarrowfromtop} = [initial above, initial text={}]
\tikzstyle{realworldarrowfromright} = [initial right, initial text={}]
\tikzstyle{realworldarrowfrombottom} = [initial below, initial text={}]

\newcommand{\eventinfigure}[2]{\scriptsize\hspace{-1mm}\ensuremath{\begin{array}{l}pre: #1 \\ post: #2\end{array}}\hspace{-2mm}}

\newcommand{\postconditiontrivial}{\raisebox{0.5mm}{\text{\tiny /}}}
\newcommand{\postconditionpfalse}{\raisebox{0.5mm}{\text{\tiny p $\leftarrow$ $\bottom$}}}

\protected\def\DTIME{\ifmmode \mbox{\sc Dtime} \else {\sc Dtime}\xspace\fi}
\protected\def\NTIME{\ifmmode \mbox{\sc Ntime} \else {\sc Ntime}\xspace\fi}
\protected\def\DSPACE{\ifmmode \mbox{\sc Dspace} \else {\sc Dspace}\xspace\fi}
\protected\def\NSPACE{\ifmmode \mbox{\sc Nspace} \else {\sc Nspace}\xspace\fi}
\protected\def\NP{\ifmmode \mbox{\sc NP} \else {\sc NP}\xspace\fi}
\protected\def\classAP{\ifmmode \mbox{\sc AP} \else {\sc AP}\xspace\fi}
\protected\def\coNP{\ifmmode \mbox{\sc coNP} \else {\sc coNP}\xspace\fi}
\protected\def\NPSPACE{\ifmmode \mbox{\sc NPspace} \else {\sc NPspace}\xspace\fi}
\protected\def\PSPACE{\ifmmode \mbox{\sc Pspace} \else {\sc Pspace}\xspace\fi}
\protected\def\APSPACE{\ifmmode \mbox{\sc APspace} \else {\sc APspace}\xspace\fi}
\protected\def\EXPSPACE{\ifmmode \mbox{\sc Expspace} \else {\sc Expspace}\xspace\fi}
\protected\def\TWOEXPSPACE{\ifmmode \mbox{\sc 2Expspace} \else {\sc 2Expspace}\xspace\fi}
\protected\def\PTIME{\ifmmode \mbox{\sc P} \else {\sc P}\xspace\fi}
\protected\def\NPTIME{\ifmmode \mbox{\sc NP} \else {\sc NP}\xspace\fi}
\protected\def\EXPTIME{\ifmmode \mbox{\sc Exptime} \else {\sc
    Exptime}\xspace\fi}
\protected\def\AEXPTIME{\ifmmode \mbox{\sc Aexptime} \else {\sc Aexptime}\xspace\fi}
\protected\def\NEXPTIME{\ifmmode \mbox{\sc NExptime} \else {\sc NExptime}\xspace\fi}
\protected\def\2EXPTIME{\ifmmode \mbox{\sc 2-Exptime} \else {\sc
		2-Exptime}\xspace\fi}
\DeclareRobustCommand{\kEXPTIME}[1][k]{\ifmmode \mbox{\sc $#1$-Exptime}
	\else {\sc $#1$-Exptime}\xspace\fi}
\DeclareRobustCommand{\kNEXPTIME}[1][k]{\ifmmode \mbox{\sc $#1$-NExptime}
	\else {\sc $#1$-NExptime}\xspace\fi}
\DeclareRobustCommand{\kEXPSPACE}[1][k]{\ifmmode \mbox{\sc $#1$-Expspace}
	\else {\sc $#1$-Expspace}\xspace\fi}
\protected\def\ELEMENTARY{\ifmmode \mbox{\sc Elementary} \else {\sc Elementary}\xspace\fi}
\protected\def\AEXPpol{\ifmmode \mbox{{\sc A}_{\text{pol}}\EXPTIME} \else
	{\sc A}$_{\text{pol}}$\EXPTIME\fi}
\protected\def\APTIME{\ifmmode \mbox{\sc Aptime} \else {\sc Aptime}\xspace\fi}
\protected\def\AEXPSPACE{\ifmmode \mbox{\sc Aexpspace} \else {\sc Aexpspace}\xspace\fi}

\newcommand{\problemdefinition}[2]{\begin{itemize}
		\item Input: #1
		\item Output: #2\end{itemize}}

\newcommand{\gloups}[1]{\bigcirc}

\tikzstyle{cell} = [draw,minimum height=5mm,minimum width=5mm]
\tikzset{
	cellcolor/.cd,
	0/.style={fill=gray!20!white},
	1/.style={fill=yellow!20!white}
}
\tikzstyle{cellalive} = [fill=yellow!20!white]

\newcommand{\automaton}{{\mathcal A}}

\newcommand{\ie}{i.e.,\ }
\newcommand{\eg}{e.g.\ }
\newcommand{\init}{{\iota}}

\newcommand{\fontforfinitedomainvariables}[1]{\mathsf{#1}}
\newcommand{\variableturn}{\fontforfinitedomainvariables{turn}}

\newcommand{\fplay}{\rho}
\newcommand{\iplay}{\pi}

\newcommand{\last}{\mbox{last}}

\protected\def\path{\ifmmode \fplay \else path\xspace\fi}
\protected\def\ipath{\ifmmode \iplay \else path\xspace\fi}

\newcommand{\strat}{\sigma}

\newcommand{\eventlettremajuscule}{A}
\newcommand{\modelM}{\kripkemodel}
\newcommand{\pointedepistemicmodel}{\modelM, \world}

\newcommand{\seteventsfor}[1]{\setevents_{#1}}
 
\newcommand\playerexists[1]{#1^{\exists}}

\newcommand{\N}{\mathbb N}
\newcommand{\epsmodel}{\kripkemodel}
\newcommand{\relworldsi}[1][\agent]{\epistemicrelation{#1}}
\newcommand{\releventsi}[1][\agent]{\epistemicrelationevents{#1}}
\newcommand{\agents}{\agtset}
\newcommand{\valworlds}{\valuationfunction}

\newcommand{\postval}{\valworlds}%\newcommand{\postval}{\val_{\text{post}}}

\newcommand{\hist}{h}

\newcommand{\ETLforest}[3][*]{{#2}{#3}^{#1}}

\newcommand{\etlmodel}{\mathcal G}
\newcommand{\trans}{\Delta}
\newcommand{\val}{\valuationelement}

\newcommand{\play}{\pi}
\newcommand{\seteventscon}{\mathsf \eventlettremajuscule_{\text{ctr}}}
\newcommand{\seteventsenv}{\mathsf \eventlettremajuscule_{\text{env}}}
\newcommand{\Play}[1]{\text{Plays}_{#1}}
\newcommand{\Hist}[1]{\text{Hist}_{#1}}

\newcommand{\is}[2]{(#1{=}#2)}

\begin{document}

\maketitle

\begin{abstract}
  We define reachability games based  on Dynamic Epistemic Logic (DEL),
where the players' actions are finely described %elaborate objects specified by action
as DEL action  models. 
 We first consider
  the setting where an external controller  with perfect information interacts
  with an environment and aims at
  reaching some epistemic goal state  regarding the passive agents of the system. We
study the problem of   strategy existence for the controller,
which generalises the classic epistemic planning problem, and we solve
it  for several types of  actions such as public announcements and
public actions. We then consider a yet richer setting where
  agents themselves are players, whose strategies must be based on
  their observations. We establish several (un)decidability results for the problem
  of existence of a distributed strategy, depending on the type of
  actions the players can use, and relate them to  results from the
  literature on multiplayer
  games with imperfect information.
  % We show that for 
% public actions this
% We show that the complexity
  % is unchanged for public announcements or public
  % actions, but the problem is undecidable for propositional
  % actions. Still, decidability is regained by restricting to the
  % so-called ``hierarchical'' knowledge among the players.
  % , as
  % considered for multi-player games with imperfect-information.
\end{abstract}

\section{Introduction}
%======================
\label{section:introduction}
\newcommand{\numberofsteps}{k}

%\fs{revoir et vendre d'abord les appli sans DEL}
%\fs{puis vendre DEL comparé à SL ou ATL etc. car l'avantage c'est quand on peut décrire les actions comme en planif}

%\fs{dire dans l'intro que ça généralise aussi \emph{contigent planning} mais avec de l'epistemic}

Many applications fall within the scope of reachability games with imperfect
information, such as video games~\cite{DBLP:conf/fun/CoulombeL18} (Civilization, etc.),
\emph{Kriegspiel} (the
epistemic variant  of
Chess)~\cite{DBLP:journals/geb/Matros18}, %pas la place de citer (),
Hanabi~\cite{baffier2016hanabi}, or contingent and conformant
planning~\cite{geffner2013concise}. % ,  to name a
% few.
%extension of contingent and conformant planning, epistemic variants of existing games such as Chess
%
%
%
% This contrasts with settings such as strategy logics \todo{citer}(ATL, SL, etc.), in which actions are described directly in the models 

\begin{table} %c'est pas une bonne idée de faire sur deux colonnes. Ca crée plein de colunnes vides
	\newcommand{\newresult}{\cellcolor{lightgray!50!white}}
	\begin{center}
          \scalebox{.8}{
            \def\arraystretch{1.3}
		\begin{tabular}{|c|c|c|c|c|}
			\hline
	& Public  announcements & %Boolean pre,
                                            Public actions  &
                                                              Propositional actions & Full \\
	% 	& announcements & %no post & 
	% actions & actions & \\
			\hline			
%\multirow{2}{*}
{Plan}   & NP-c & %PSPACE-c &
		 PSPACE-c  & decidable                                                      
                        & undecidable
                  \\
             %     &  \cite{bolander2015complexity,PhDTristanCharrier}
              %                            & \cite{PhDTristanCharrier}
               %            &  \cite{DBLP:conf/ijcai/YuWL13,DBLP:journals/corr/AucherMP14,AiML2018Gaetanetal}
                                     %&  %\cite{DBLP:journals/jancl/BolanderA11,DBLP:conf/ijcai/CongPS18}\\
					\hline		
Controller   & \newresult{PSPACE-c} (Th.~\ref{theorem:epistemicgameproblem-publicannouncements-PSPACEcomplete})  &  %\atrouver{EXPTIME-c} & 
	     \newresult	EXPTIME-c (Th.~\ref{theorem:epistemicgameproblem-publicactions-EXPTIMEcomplete}) & \newresult
                                                       decidable
                                                       (Th.~\ref{theorem:epistemicgameproblem-prepostBoolean-decidable}) &
                                                                     undecidable \\ \hline
                  % & \newresult  Th.~\ref{theorem:epistemicgameproblem-publicannouncements-PSPACEcomplete}   &   \newresult Th.~\ref{theorem:epistemicgameproblem-publicactions-EXPTIMEcomplete}   & \newresult Th. \ref{theorem:epistemicgameproblem-prepostBoolean-decidable} & \\

   & \newresult & \newresult & \newresult      undecidable  (Th.~\ref{theorem:uniformstrategyundecidable}) &    \\
 % & \newresult Th.~\ref{theorem:uniformstrategyexistence-publicannouncements-PSPACEcomplete} &  \newresult  Th.~\ref{theorem:uniformstrategyexistence-publicactions-EXPTIMEcomplete} & \newresult Th.~\ref{theorem:uniformstrategyundecidable}   &        \\
\multirow{-2}{*}{Distributed strategy} &  \newresult
                             \multirow{-2}{*}{PSPACE-c
                             (Th.~\ref{theorem:uniformstrategyexistence-publicannouncements-PSPACEcomplete})}
                               & \newresult  \multirow{-2}{*}{EXPTIME-c (Th.~\ref{theorem:uniformstrategyexistence-publicactions-EXPTIMEcomplete})}  & \newresult decidable case (Th.~\ref{theorem:uniformstrategydecidable})  & \multirow{-2}{*}{undecidable}\\
   		\hline
		\end{tabular}}	
	\end{center}
	\caption{Known  results (new in grey) for  plan,
          controller and distributed strategy synthesis.}
        \label{table:results}
\end{table}

Games with imperfect information are computationally hard, and 
even undecidable for multiple
players~\cite{DBLP:conf/focs/PetersonR79}.
% Epistemic variants of strategic
% logics~\cite{DBLP:journals/iandc/BerwangerCWDH10,DBLP:conf/fossacs/Bouyer18,DBLP:conf/ijcai/BelardinelliLMR17} % such as Alternating-Time Temporal Logic \cite{} or Strategy Logic \cite{}
% %
% allow reasoning about existence of strategies in such games, but
% complexities are high.
One way to tame this complexity is to make assumptions on how the
knowledge of the different players compare: if all players that cooperate can be ordered in a hierarchy where one knows more than
the next, a
situation  called \emph{hierarchical information}, then the existence of distributed strategies can be
decided~\cite{peterson2002decision,DBLP:journals/acta/BerwangerMB18}. 
Another natural approach is to consider fragments based on classes of
action types, as done  for instance
in~\cite{ramanujam2010communication,DBLP:conf/atal/BelardinelliLMR17,DBLP:conf/fossacs/Bouyer18}
where different kinds of public actions are considered. But the usual
graph-based models of games with imperfect information, where the
players' actions are modelled as labels on the edges,
make it difficult to define subtle properties of actions.

%  public/private announcements, public actions, etc. (note however attempts to define public actions in concurrent game structures \cite{DBLP:conf/atal/BelardinelliLMR17}).

By contrast, Dynamic epistemic logic (DEL) \cite{DitmarschvdHoekKooi}
was designed to describe actions precisely: how they
affect the world and how
they are perceived. % DEL provides a compact
% representation of actions that extend STRIPS, called \emph{event
%   models}.
In particular, classic action types %\bm{considered where? \sp{cannot concisely answer but with the DEL book, what do you think?}}
% \cite{DitmarschvdHoekKooi}
such as public/private
announcements or public actions
 correspond to natural classes of DEL \emph{action models}. Also, DEL
extends epistemic logic and hence enables  modelling
\emph{higher-order knowledge}, \ie what an agent knows about what another agent
knows etc, and the evolution of agents' knowledge over  time.
%
%\fs{dire que ya des jeux bornés (jeux de cartes, Hanabi, etc.), mais que là on s'en fout}

%DEL enables to reason about bounded games, that is, games on which we
%can compute in advance a bound of the number of steps $\numberofsteps$
%to reach a final state. Typical bounded games are card games: Belote,
%Hanabi, to cite a few (the game is finished when all cards are played,
%therefore take $\numberofsteps$ to be the total number of cards). In
%DEL, we typically write a formula where we alternate $\numberofsteps$
%times an existential and universal action modality
%$\ldiaarg{\setactions_1} \lboxarg{\setactions_2}\dots
%\ldiaarg{\setactions_1} \lboxarg{\setactions_2}\psi$, where
%$\setactions_i$ is the set of actions for player $i$ to express the
%existence of a winning strategy leading to a $\psi$-state. Thus,
%reasoning about bounded epistemic games in DEL reduces to model
%checking, and is PSPACE-complete even for succinct languages
%\cite{DBLP:conf/atal/CharrierS17,AiML2018DELCK}.

A classification of the complexity with respect to action types was
addressed in the literature of \emph{epistemic planning}, a problem that asks
for the existence of a \emph{plan}, \ie a finite sequence of DEL
actions to reach a situation that satisfies some given objective
expressed in epistemic logic.
 % proven to be undecidable \cite{DBLP:journals/jancl/BolanderA11}, that can be seen as a one-player epistemic game where the unique player still reason about the knowledge of the other agents.
%
%
%
However this problem, which can be seen as solving one-player reachability
games with epistemic objective,  has never been
considered in a strategic, adversarial context. 
This work bridges the gap between DEL and games by introducing
adversarial aspects in  DEL planning,  thus
moving from plan generation to strategy synthesis. We define two
frameworks for DEL-based
reachability games, where players start in a given epistemic situation
and their possible moves are described by action models, % . As in
% epistemic planning,
and the objective is  to reach a situation
 satisfing some
epistemic formula.

%François : j'ai supprimé le paragraphe là car :
%- je n'aime pas trop dire un "défaut de DEL" (surtout que c'est un petit défaut, c'était pas si dur à résoudre ça)
%- on a peu de place
%One main reason is
%that DEL models how agents observe the occurrence of events or
%actions, but not who performs these actions. However bridging the gap
%between DEL and games would allow a much finer representation of
%player's moves and how they are perceived than what is allowed by the
%usual finite graphs usually considered in games with imperfect
%information~\cite{DBLP:journals/iandc/BerwangerCWDH10,DBLP:conf/fossacs/Bouyer18}
%and multi-agent 
%logics for strategic reasoning with imperfect
%information~\cite{BMMRV17,DBLP:conf/ijcai/BelardinelliLMR17}.
%\bm{cite more?}\fs{en parler AVANT DEL (au début de l'intro) ?}

In a first step we consider 
 \emph{open systems}~\cite{harel1985development}, \ie systems that
 interact with an environment. In our setting, two
omniscient, external entities (that we call \emph{controller} and
\emph{environment}) choose in turn which actions are performed. We call
this setting \emph{DEL controller synthesis}. Here, agents involved in the
models and formulas are not active, they merely observe how the system
evolves based on the actions chosen by the controller and the
environment, and update their knowledge accordingly. DEL controller synthesis
extends DEL planning, as the latter is a degenerate case of the former where
the environment stays idle, and we therefore inherit undecidability for the general case.
Nevertheless we show that, as for DEL planning, decidability is regained when
actions do not increase uncertainty (so-called \emph{non-expanding actions}) or when the preconditions of actions are
propositional formulas. More precisely, we show \PSPACE-completeness
when possible moves are public announcements, \EXPTIME-completeness
for the more general public actions, and membership in \kEXPTIME[(k+1)]
for propositional actions when the objectives are formulas  of modal
 depth at most $k$.

% As far as we are concerned, such fine-grained
                   %results are new. \bm{bien sûr qu'on présente des
                   %résultats nouveaux}

We then generalise further this setting by turning agents into
players. Unlike the omniscient controller of the former setting,
agents have imperfect information about the current state of the game,
and can only base their decisions on what they know. In the theory of
games with imperfect information this is modelled by the notion of
\emph{uniform strategies}, also called \emph{observation-based
  strategies}~\cite{apt2011lectures}. % When agents have perfect recall
% of the past, which is the case in DEL
% (see~\cite{DBLP:conf/tark/DegremontLW11}), imperfect information makes
% game solving much harder: it is already undecidable when two agents
% with incomparable knowledge try to cooperate to enforce a reachability
% objective against an opponent~\cite{DBLP:conf/focs/PetersonR79}.
% However, games with imperfect information are known to be solvable in
% two main  cases: either when 
% knowledge among the players in the coalition is \emph{hierarchical}, meaning that there is a hierarchy
% between agents in the coalition such that each agent knowns more than
% those below
% her~\cite{peterson2002decision,PR90,kupermann2001synthesizing}, or
% when all events are public~\cite{DBLP:conf/concur/MeydenW05,ramanujam2010communication,DBLP:conf/ijcai/BelardinelliLMR17,DBLP:conf/fossacs/Bouyer18}.
We study the problem of \emph{distributed strategy synthesis}, where a
group of players cooperate to enforce some objective against the
remaining players.
As for multi-player games with imperfect information the problem
  is undecidable, already for
propositional actions and a coalition of two players. However 
we show
that the two kinds of assumptions that make imperfect-information
games decidable, namely public actions and hierarchical information,
also yield decidable cases of multiplayer DEL games. %\bm{not richer!}
Furthermore, in the case of
public announcements and public actions, the complexity
is not worse than for controller synthesis.

\paragraph*{Related work}
The complexity of DEL-based epistemic planning has been 
thoroughly investigated. It is undecidable already for actions with
preconditions of modal depth one and propositional postconditions
\cite{DBLP:journals/jancl/BolanderA11,DBLP:conf/ijcai/CongPS18}. 
For preconditions of modal depth one and no postconditions the problem
has been open for years, but it is decidable when pre- and postconditions are
propositional~\cite{DBLP:conf/ijcai/YuWL13,DBLP:journals/corr/AucherMP14,AiML2018Gaetanetal}.
It is also known to be \NP-complete for
 public announcements~\cite{bolander2015complexity,PhDTristanCharrier},
%  (the membership in NP is actually stated in
% \cite{bolander2015complexity} only when preconditions are
% propositional, but their proof trivially extends for arbitrary
% preconditions and details can be found in \cite{PhDTristanCharrier},
% p.~101).
and \PSPACE-complete for
public actions~\cite{PhDTristanCharrier}. % . PSPACE-hardness follows from the
% PSPACE-hardness of classical planning. PSPACE-membership is yield from
% a non-deterministic algorithm that successively chooses an event to
% trigger, the size of the current epistemic model is always smaller or
% equal to the size of the initial epistemic model, the goal can be
% checked in polynomial time (see details in

The decidability for propositional actions has been extended 
in~\cite{DBLP:journals/corr/AucherMP14} by considering infinite trees
of actions called \emph{protocols} instead of finite plans, and  specifications in branching-time epistemic
temporal logic instead of reachability for epistemic formulas; this
has been extended further in~\cite{AiML2018Gaetanetal} by enriching
the specification language with Chain Monadic Second-order Logic.
Both results rely on the fact that when actions are propositional,
the infinite structures generated by  repeated application of action models
form a class of \emph{regular
  structures}~\cite{DBLP:journals/corr/AucherMP14,DBLP:phd/hal/Maubert14},
\ie relational structures that have a finite representation
via automata. First-order logic is decidable
on such structures~\cite{blumensath2000automatic}, and chain-MSO is decidable
on a subclass called \emph{regular automatic
  trees}~\cite{AiML2018Gaetanetal}, but neither of these logics can express the existence of
strategies in games. However we will show that  the regular structures
obtained from propositional DEL models can be seen as finite turn-based game
arenas studied in games played on graphs. This allows us to transfer  decidability results on
games with epistemic temporal objectives to the DEL setting.

A notion of cooperative planning in DEL has been studied
in~\cite{DBLP:journals/corr/EngesserBMN17}, but without the
adversarial aspect of games. Also, in~\cite{DBLP:journals/logcom/Lima14}, a game
setting has been developed with the so-called Alternating-time
Temporal Dynamic Epistemic Logic, but it does not consider uniform
strategies and thus cannot express the existence of distributed strategies. Our controller synthesis problem can
be expressed in this logic, but not in the fragment that they
solve, which cannot express reachability. 
On the other hand, several decidability
results for logics for strategic and epistemic reasoning have been
established
recently~\cite{DBLP:conf/ijcai/BelardinelliLMR17,DBLP:conf/kr/MaubertM18},
but they do not offer the fine modelling of actions possible in DEL.
For instance they cannot easily model public announcements, which we
show yield better complexity  than those obtained in their
settings.

Table~\ref{table:results} sums up previous results for epistemic
planning, as well as the results established in this contribution.

%\fs{example de jeu à info imparfaite non borné : imperfect information Chess \francoisfait}

%http://delivery.acm.org/10.1145/3100000/3091301/p1268-belardinelli.pdf?ip=131.254.66.168&id=3091301&acc=ACTIVE%20SERVICE&key=7EBF6E77E86B478F%2E9BD6B3DBCD4B0A3B%2E4D4702B0C3E38B35%2E4D4702B0C3E38B35&__acm__=1550495452_9ba3b7f110bb84205f5f15ac4efd6c23

% \fs{l'existence d'une stratégie uniforme codés en FO sur structures automatiques : pas assez expressif}

% \fs{l'existence d'une stratégie uniforme codés en MSO sur structures automatiques : possible mais indécidable}

% \fs{avec CTL*, oui}

% \fs{mais ya un framework proche : celui de Bastien mais il ne considérait pas les actions}

\newcommand{\atrouver}[1]{\textcolor{green!80!black}{#1?}}

\section{Background in epistemic planning}
%======================
\label{section:backgroundepistemicplanning}

Let us fix a countable set of \emph{atomic propositions} $\atmset$. %   to
 % denote facts about the world.

\subsection{The classic DEL setting}%Epistemic Logic}
\label{sec-EL}
%\paragraph{Epistemic Logic}
% -----------------------

We recall models of % \emph{epistemic models} from
epistemic logic~\cite{Fagin95knowledge}. 

\label{sec-defDEL}
\begin{definition}%[Epistemic model]
	\label{definition:epistemicmodel}
	An \indexemph{epistemic model} $\kripkemodel = (\setworlds,
        (\epistemicrelation{\agent})_{\agent \in \agtset},
        \valuationfunction)$ is a tuple  where %:
	\begin{itemize}
		\item
 $\setworlds$ is a non-empty finite set of possible \emph{worlds} (or situations),
		\item
 $\relworldsa\subseteq \setworlds\times \setworlds$
		is an
		\emph{accessibility relation} for agent $\agent$, and
		\item
 $\valuationfunction:\setworlds\rightarrow 2^\atmset$ is a \emph{valuation function}.
	\end{itemize}
\end{definition}

\tikzstyle{agenta}=[line width=2]
\tikzstyle{agentb}=[dashed]

We write $\world \relworldsa \worldb$ instead of
$(\world,\worldb)\in\relworldsa$; its intended meaning is that when
the actual world is $\world$, agent~$\agent$ considers that $\worldb$
may be the actual world.  The valuation function $\valuationfunction$
provides the subset of atomic propositions that hold in a world.  A
pair $(\kripkemodel,\world)$ is called a \indexemph{pointed epistemic
  model}, and we let $|\kripkemodel|$ be the \emph{size} of
$\kripkemodel$, defined as
$|\setworlds| + \sum_{\agent \in \agtset} |\epistemicrelation{\agent}|
+ \sum_{\world \in \setworlds} |\valuationfunction(\world)|$. We will
only consider finite models, \ie we assume that
$\valuationfunction(\world)$ is finite for all worlds. % Note that
% when an epistemic model is given as an input, we also suppose that
% $\valuationfunction(\world)$ is finite for all worlds
% $\world \in
% \setworlds$. % or simply a \indexemph{state}. We also use the notation $\astate$ to denote a state.
%We call \emph{$\val$-world}, where $\val\subseteq \AP$, a world $\world$ such that $\valuationfunction(\world) = \val$.

The syntax of Epistemic Logic $\languageEL$ is given by the following grammar:
\[\phi\grammaris \atom \grammarsep \neg\phi \grammarsep (\phi\lor\phi) \grammarsep \lknow\agent\phi\]
where $\atom$ ranges over $\atmset$ and $\agent$ ranges over
$\agtset$.

$\lknow\agent\phi$ is read `agent $\agent$ knows that
$\phi$ is true'. We define the usual abbreviations
$(\phi_1 \land \phi_2)$ for $\lnot (\lnot \phi_1 \lor \lnot \phi_2)$
and $\lknowpos a \phi$ for $\lnot \lknow a \lnot \phi$, and use 
% The size of a formula $\phi$ is the number of operators needed to write $\phi$.
$\languagePropositional$ for the fragment of $\languageEL$ with 
propositional formulas only. The \emph{modal depth} of a formula is its maximal number of
nested knowledge operators; for instance, the formula
$\lknow \agent \lknow \agentb p \land \lnot \lknow a q$ has modal depth $2$.
The \emph{size} $|\phi|$ of a formula $\phi$ is the number of symbols
in it.

The semantics of $\languageEL$ relies on \emph{pointed epistemic
  models}. 
\begin{definition}
  We define $\kripkemodel, \world\models \formula$, read as
`formula $\formula$ holds in the pointed epistemic model
$(\kripkemodel, \world)$', by induction on $\phi$, as follows:
%	$\kripkemodel, \world\models \formula$ is defined by induction on $\formula$:
	\begin{itemize}
		\item $\kripkemodel,\world\models \atom$ if $\atom\in\valuationfunction(\world)$;
		\item $\kripkemodel,\world\models \neg\phi$ if it is not the case that $\kripkemodel,\world\models\phi$;
		\item $\kripkemodel,\world\models (\phi\lor\psi)$ if
		$\kripkemodel,\world\models\phi$ or $\kripkemodel,\world\models\psi$;
\item $\kripkemodel,\world\models\lknow\agent\phi$ whenever for all $\worldb$ such that $\world\relworldsa\worldb$, $\kripkemodel,\worldb\models\phi$.
	\end{itemize}
\end{definition}

%\subsection{Event models}
%\label{sec-event-models}
%\paragraph{Action models}
%-----------------------
Dynamic Epistemic Logic (DEL) relies on \emph{action models} (also
called ``event models'').  These models
specify how agents perceive the occurrence of an action as well as its effects on the world.

\begin{definition}%[Action model]
	\label{def-actionmodel}
	An \emph{action model}
	$\eventmodel = \eventmodeltuple$ is a tuple  where:
	\begin{itemize}
		\item
 $\setevents$ is a non-empty finite set of  possible \emph{actions},
		\item
 $\epistemicrelationevents\agent\subseteq\setevents\times \setevents$
		is the \emph{accessibility
			relation}  for agent $\agent$,
		\item
 $\pre: \setevents\to \languageEL$ provides the \emph{precondition}
 for an action to be performed, and
		\item
		 $\post: \setevents \times \atmset \to
                 \languagePropositional$ provides the
                 \emph{postcondition} (\ie the effects) of an action.
	\end{itemize}
\end{definition}

A \emph{pointed action model} is a pair $(\eventmodel,\event)$
where $\event$ represents the actual action. % A pair
% $(\eventmodel, \multipoint)$ with $\multipoint \subseteq \setevents$
% is called a \emph{multi-pointed action model}, where $\multipoint$
% represents the set of possible actual actions.
We let $|\eventmodel|$ be the \emph{size} of
$\eventmodel$, defined as
$|\eventmodel|\egdef |\setevents|+\sum_{\agent \in \agtset} |\epistemicrelationevents\agent| + \sum_{\event\in\setevents}|\pre(\event)|+\sum_{\event\in\setevents,p\in\AP}|\post(\event,p)|$.
% its number of actions plus the sum of the
% sizes of all its pre- and postconditions.

% Pointed action models correspond to deterministic actions and multi-pointed actions models correspond to non-deterministic actions. We identify $(\eventmodel,\event)$ and $(\eventmodel,\set{\event})$.
% A pointed event model is called an \emph{action}\footnote{We use the word \emph{action} although according to \cite{audi_2015}, an action should always have an author.} and is denoted by $\action$. \sp{là c'est franchement la cata des notations}

An action $\event$ is \emph{executable} in a world $\world$ of an
epistemic model $\epsmodel$ if
%its precondition $\pre(\event)$ holds in $\world$, i.e.
$\kripkemodel,\world\models\pre(\event)$, and in that case we define
$\postval(\world,\event):=\set{\atom \in \atmset \suchthat
  \kripkemodel, \aworld \models \post(\event, \atom) }$, the set of
atomic propositions that hold after occurrence of action $\event$ in
world $\world$.
Since postconditions are always propositional, we can define
similarly $\postval(\val,\event)$ where $\val\subseteq 2^\atmset$ is a valuation.

\paragraph*{Types of actions} We identify noticeable types of
actions. An action model $\eventmodel$ is \emph{propositional} if all
pre- and postconditions of actions in $\eventmodel$ belong to
$\languagePropositional$. A \emph{public action} is a pointed action
model $\eventmodel,\event$ such that for each
agent $\agent$, $\epistemicrelationevents{\agent}$ is the identity relation.
% all its action $\event'$, for all agents $\agent$, $\event \epistemicrelationevents{\agent} \event'$ implies $\event = \event'$.
A \emph{public announcement} is a public
action $\eventmodel,\event$ such that for all $p$,
$\post(\event, p) = p$. % , which we call a \emph{trivial postcondition}.
% ; we abuse vocabulary and say that it has \emph{no postcondition}.
%A public announcement $\eventmodel,\event$ with precondition $\phi$ is called an announcement of~$\phi$.

% After occurrence of an event $\event$ in a world $\world$, agent
% $\agent$ considers it possible that event $\event'$ occured
% in world $\world'$ if in $\world$ she considers $\world'$
% possible, $\event'$ is executable in $\world'$, and she
% considers event $\event'$ possible when event $\event$ is executed.
% This leads to the following definition of the
We recall the product that models how to update an epistemic model
when an action is executed~\cite{baltag1998logic}.

\begin{definition}[Product~\cite{baltag1998logic}]
	\label{definition:product}
	Let
        $\kripkemodel = (\setworlds,
        (\epistemicrelation\agent)_{\agent \in \agtset},
        \valuationfunction)$ be an epistemic model, and\linebreak[4]
        $\eventmodel = (\setevents,
        (\epistemicrelationevents{\agent})_{\agent \in \agtset},\pre,
        \post)$ be  an action model. The \emph{product} of
        $\kripkemodel$ and $\eventmodel$ is defined as
        $\kripkemodel \otimes \eventmodel = (\setworlds',
        (\epistemicrelation\agent)', \valuationfunction')$ where:
	\begin{itemize}
		\item $\setworlds' = \set{(\aworld, \event) \in \setworlds \times \setevents \suchthat \kripkemodel, \aworld \models \pre(\event)}$,
		\item $(\aworld,\event) \epistemicrelation\agent'
                  (\aworld',\event')$ if $\aworld
                  \epistemicrelation\agent \aworld'$ and $\event
                  \epistemicrelationevents{\agent} \event'$, and
		\item $\valuationfunction'((\aworld, \event)) = \postval(\world,\event)$.
	\end{itemize}
\end{definition} 

% The product is extended to pointed epistemic models and pointed action models in a natural manner: $(\kripkemodel,\world)\otimes
% (\eventmodel,\event)\egdef(\kripkemodel\otimes\eventmodel, (\world,\event))$, that actually exists only if $\kripkemodel,\world\models\pre(\event)$ (i.e.\ action $\event$ can be executed in world $\world$).

% The product of a pointed epistemic model $(\kripkemodel,\world)$ with a
% pointed action model $(\eventmodel,\event)$ is defined only if
% $\kripkemodel,\world\models\pre(\event)$ (i.e.\ action $\event$ can be executed in world $\world$) and is 
% the epistemic model
% $(\kripkemodel,\world)\otimes
% (\eventmodel,\event)\egdef(\kripkemodel\otimes\eventmodel, (\world,\event))$.

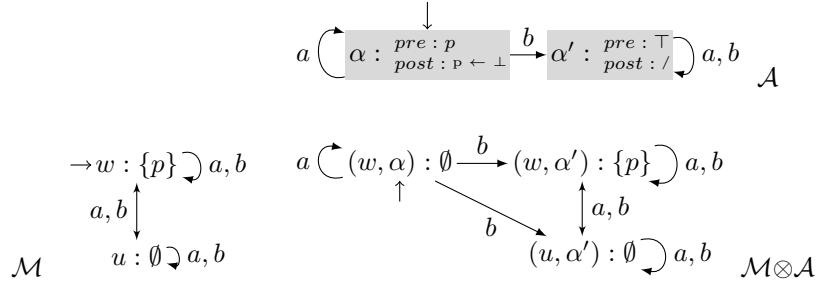
\begin{figure}
	\begin{center}

		%action model
		\scalebox{1}{
		\begin{tabular}{ll}
			&
			\begin{tikzpicture}[scale=0.8]%event model
			\tikzstyle{double_edge} = [latex'-latex',double]
			\node[event, realworldarrowfromtop] (p) at (0,0) {$\event: \eventinfigure{p} \postconditionpfalse$};
			\node[event] (notp) at (3,0) {$\eventb: \eventinfigure{\top} \postconditiontrivial$};
			\draw[-latex] (p) edge  node[above] {$\agentb$} (notp);
			\draw[-latex] (p) edge[loop left, looseness=2] node[left] {$\agent$} (p);
			\draw[-latex] (notp) edge[loop right, looseness=2] node[right] {$\agenta, \agentb$} (notp);
			\end{tikzpicture} $\eventmodel$ \\[4mm]
			
			$\kripkemodel$
			\begin{tikzpicture}[scale=0.6]
			\tikzstyle{double_edge} = [latex'-latex',double]
			\node[world, realworldarrowfromleft] (p) at (0,0) {$\world: \{p\}$};
			\node[world] (notp) at (0,-2) {$\worldb: \emptyset$};
			\draw[double_edge] (p) edge  node[left] {$\agenta, \agentb$} (notp);
			\draw[-latex] (p) edge[loop right,  looseness=2.5] node[] {$\agenta, \agentb$} (p);
			\draw[-latex] (notp) edge[loop right,  looseness=2.5] node[] {$\agenta, \agentb$} (notp);
			\end{tikzpicture}

			&

			\begin{tikzpicture}[xscale=0.8, yscale=0.6]
			\tikzstyle{double_edge} = [latex'-latex',double]
			\node[world, realworldarrowfrombottom] (p) at (0,0) {$(\world, \event): \emptyset$};
			\node[world] (p2) at (3,0) {$(\world, \eventb): \{p\}$};
			\node[world] (notp) at (3,-2) {$(\worldb, \eventb): \emptyset$};
			\draw[-latex] (p) edge  node[above] {$\agentb$} (p2);
			\draw[-latex] (p) edge  node[below] {$\agentb$} (notp);
			\draw[-latex] (p) edge[loop left, looseness=2] node[left] {$\agent$} (p);
			\draw[-latex] (notp) edge[loop right, looseness=2] node[right] {$\agenta, \agentb$} (notp);
			\draw[-latex] (p2) edge[loop right, looseness=2] node[right] {$\agenta, \agentb$} (p2);
			\draw[double_edge] (p2) edge  node[right] {$\agenta, \agentb$} (notp);
			\end{tikzpicture}
			$\kripkemodel {\otimes} \eventmodel$
		\end{tabular}}
	\end{center}
	\caption{Example of DEL product. Symbol \;$\postconditiontrivial$\; indicates
        the trivial postcondition that leaves valuations unchanged.}
	\label{figure:productexample}
\end{figure}

\begin{example}
  Figure \ref{figure:productexample} shows the pointed model
  $\kripkemodel, \world$ that represents a situation in which $p$ is
  true and both agents $\agent$ and~$\agentb$ do not know it. The
  pointed action model $\eventmodel, \event$ describes the action where
  agent $\agent$ learns that $p$ was true but that it is now set to
  false, while agent $\agentb$ does not learn anything (she sees action
  $\eventb$ that has trivial pre- and postcondition). In the product 
epistemic model $\kripkemodel \otimes \eventmodel, (w, \event)$,
  agent $\agent$ now knows that $p$ is false, while
  $\agentb$ still does not know the truth value of $p$, or whether agent $a$ knows it.
\end{example}

An epistemic model (resp.\ an action model) is \emph{S5} if all
accessibility relations are equivalence relations. This property is
important to model games with imperfect information, and we will
assume it in Section~\ref{section:imperfectinfoepistemicgames}.

% In the sequel, given an event model
% $\eventmodel = (\setevents, (\epistemicrelationevents{\agent})_{\agent
%   \in \agtset},\pre, \post)$ and a subset
% $\subsetevents\subseteq\setevents$ of events, we introduce the modal
% construction $\ldiaarg{\subsetevents}\phi$, read as `there is an event
% $\event$ in $\subsetevents$ that is executable and $\phi$ holds after
% having executed it' and its dual construction
% $\lboxarg{\subsetevents}\phi$, read as `for all events $\event$ in
% $\subsetevents$, if $\event$ is executable then $\phi$ holds after
% having executed it'. The semantics is (for a given event model $\eventmodel$):
% \begin{itemize}
% 	\item $\kripkemodel, \world \models \ldiaarg{\subsetevents}\phi$ if there exists $\event \in \subsetevents$ such that  $\kripkemodel, \world \models\pre(\event)$ and $\kripkemodel \otimes \eventmodel, (\world, \event) \models \phi$;
% 	\item $\kripkemodel, \world \models \lboxarg{\subsetevents}\phi$ if for all $\event \in \subsetevents$, $\kripkemodel, \world \models\pre(\event)$ implies $\kripkemodel \otimes \eventmodel, (\world, \event) \models \phi$.	
% \end{itemize}

\subsection{Generated structure}
\label{sec-generated}
% --------------------------

Iteratively executing an action model from an initial epistemic model
generates an infinite sequence of epistemic models, whose union yields an infinite epistemic structure where dynamics are
represented by the possible sequences of actions, while information is
captured by the accessibility relations.

\begin{definition}[Generated structure] \index{DEL-generated model}\index{$\ETLforest{\epsmodel}{\eventmodel}$}
\label{def-DELforest}
Given
$\epsmodel=(\setworlds,\{\relworldsi\}_{\agent\in\agents},\valworlds)$
 an epistemic model and\linebreak[4]
$\eventmodel=(\setevents,\{\releventsi\}_{\agent\in\agents},\pre,\post)$
an action model,
we define the family of disjoint epistemic models $\{\ETLforest[n]{\epsmodel}{\eventmodel}\}_{n\geq 0}$
by letting $\ETLforest[0]{\epsmodel}{\eventmodel}=\epsmodel$ and
$\ETLforest[n+1]{\epsmodel}{\eventmodel}=\ETLforest[n]{\epsmodel}{\eventmodel}\otimes\eventmodel$.
We finally define the infinite epistemic model
$\ETLforest{\epsmodel}{\eventmodel}=\bigcup_{n\in\N}\ETLforest[n]{\epsmodel}{\eventmodel}$. 
% Letting,
% for each $n$,
% $\ETLforest[n]{\epsmodel}{\eventmodel}=(\setworlds^n,\{\relworldsi^n\}_{\agent\in\agents},\valworlds^n)$,
% we define the infinite epistemic model $\ETLforest{\epsmodel}{\eventmodel}=
% (\Sigma,\ETLdom,\{\ETLrel_\agent\}_{\agent\in\agents},\ETLval)$, where:
% \begin{itemize}
% \item $\Sigma=\setworlds\union\setevents$,
%   \item $\ETLdom=\bigunion_{n\geq 0}\setworlds^n$,
%   \item $\hist\ETLrel_\agent\hist'$ if there is some $n$ such that
%     $\hist,\hist'\in\epsmodel^n$ and $\hist \relworldsi^n \hist'$, and
%   \item $\ETLval(p)=\bigunion_{n\geq 0}\valworlds^n(p)$.
% \end{itemize}
\end{definition}

In the following we identify objects of the form
$(\ldots((w,\event_1),\event_2),\ldots \event_n)$ with $(w,\event_1,\ldots,\event_n)$. Anticipating the game setting we later define, we call a \emph{play}
an infinite sequence $\play=\world\event_1\event_2\ldots$ such that
all finite prefixes of $\play$ are in
$\ETLforest{\epsmodel}{\eventmodel}$. A \emph{history} is a finite
prefix $\hist$ of a play. We let
$\Play{\ETLforest{\epsmodel}{\eventmodel}}(\world)$ and
$\Hist{\ETLforest{\epsmodel}{\eventmodel}}(\world)$ be, respectively,
the set of all plays and histories in
${\ETLforest{\epsmodel}{\eventmodel}}$ that start with $\world$. % In
% these notations, we may omit the parameter $\world$ when it is some
% initial world.
These definitions entail the following.

\begin{lemma}
\label{lem-sem-DEL-ETL}
 For every world
$(\world,\event_1,\ldots,\event_n)\in \epsmodel\otimes\eventmodel^n$, 
and every formula $\phi\in\languageEL$, 
\[\epsmodel\otimes\eventmodel^n,(\world,\event_1,\ldots,\event_n)\models\phi
\mbox{ iff }\ETLforest{\epsmodel}{\eventmodel},\world\event_1\ldots\event_n\models\phi.\]
\end{lemma}

This shows that the alternative  definition of epistemic
planning given in the next section is equivalent to the usual one.

\subsection{Epistemic planning}
\label{sec-epistemic-planning}
%-----------------------

The epistemic planning problem 
%\bm{plan existence or epistemic planning?}
%plutôt plan existence problem car c'est plus uniforme avec les autres pbs de décision
%epistemic planning c'est plutôt le field
%
 asks for the existence of an executable  sequence of designated actions
$\event_1, \dots, \event_n$  in an action model
$\eventmodel$, whose execution from
$\pointedepistemicmodel$ leads to a situation satisfying some
objective expressed as an epistemic logic formula. % objective expressed in $\langEL$.
Formally, we consider the following problem.% This
% problem can be phrased as checking for the existence of an integer $n$
% such that
% $\pointedepistemicmodel \models \ldiaarg{\subsetevents}^n\phi$, where
% $\ldiaarg{\subsetevents}^n$ is
% $\ldiaarg{\subsetevents} \dots \ldiaarg{\subsetevents}$ where
% $\ldiaarg{\subsetevents}$ is repeated $n$ times.

%\bm{on est obligé de considérer $\multipoint$?}

\begin{definition}[Plan existence problem]
	\label{def-planning}
	~
	
	\problemdefinition{a pointed epistemic model $\pointedepistemicmodel$, an action model $\eventmodel$ % with a finite designated
		% subset $\multipoint \subseteq \setevents$
                and an \emph{objective} formula $\formula \in \languageEL$;}{yes if there is a
                history $\hist$ in
                $\Hist{\ETLforest{\epsmodel}{\eventmodel}}(\world)$
                % containing only actions in $\multipoint$ and 
                such that
                $\ETLforest{\epsmodel}{\eventmodel},\hist\models\formula$.}

                % integer $n$ such that 
		% $\pointedepistemicmodel \models \ldiaarg{\subsetevents}^n\phi$; no otherwise.}
	%The sequence $(\modelAction_1,\event_1),\ldots,(\modelAction_\param,\event_\param)$ is called a \emph{plan}. 
            \end{definition}

            \begin{remark}
              Note that usual formulations of the plan existence
              problem consider a set of distinct pointed action models
              $(\eventmodel_1,\event_1),\ldots,(\eventmodel_n,\event_n)$
              instead of one action model $\eventmodel$. Both
              formulations are equivalent, in the sense that they are
              interreducible in linear time.
            \end{remark}

            Main known results on the plan existence problem are
            summarised in Table~\ref{table:results}, while the
            relevant pointers to the literature are given in the
            related work paragraph of the introduction.

%, 
 
%\fs{parler de epistemic protocol synthesis}
%\fs{+ structures automatiques}
%
%\fs{+ que planning est décidable}
%
%
%
%\fs{+ énoncé du problème sur les jeux}
%
%\fs{on se concentre sur DEL propositionnel}

We recall some standard notions
            and notations that we will need in the rest of the paper.  A \emph{finite}
            (resp. \emph{infinite}) \emph{word} over some alphabet $\Sigma$ is an
            element of $\Sigma^{*}$ (resp. $\Sigma^{\omega}$).  The
            \emph{length} of a finite word $w=w_{0}w_{1}\ldots w_{n}$
            is $|w|\egdef n+1$, and $\last(w)\egdef w_{n}$ is its last
            letter.  Given a finite (resp. infinite) word $w$ and
            $0 \leq i < |w|$ (resp. $i\in\N$), we let $w_{i}$ be the
            letter at position $i$ in $w$, $w_{\leq i}:=w_0\ldots w_i$
            is the prefix of $w$ that ends at position $i$ and
            $w_{\geq i}:=w_iw_{i+1}\ldots$ is the suffix of $w$ that
            starts at position $i$.
We also use variables $x$ that range over some finite
domain. We
will write $(x = d)$ for the fact ``the value of $x$ is $d$'', and use
$x:=d$ for the effect of setting $x$ to value $d$. This can all be encoded with atomic propositions.
%independently of the world in which it is applied. 

%%% Local Variables:
%%% mode: latex
%%% TeX-master: "main"
%%% End:

\section{Controller synthesis}
%=============================
\label{section:controller}
We first generalise the plan existence problem to the setting where
some environment may perturb the execution of the plan that should
thus be robust against it.%all possible behaviours of this environment.

Formally, we consider  an initial  epistemic model
$\kripkemodel$, as in Definition~\ref{definition:epistemicmodel},
%=(\setworlds, (\epistemicrelation{\agent})_{\agent \in\agtset}, \valuationfunction)$
with an initial world $\world_\init$, and an action model
$\eventmodel = (\setevents, (\epistemicrelationevents{\agent})_{\agent
  \in \agtset},\pre, \post)$ whose set of actions
$\setevents$ is partitioned into
actions in $\seteventscon$ controlled by a \emph{Controller} and actions in $\seteventsenv$ controlled by the \emph{Environment}.

Controller and Environment play in turn: in each round, 
Controller first chooses to execute an action in
$\seteventscon$, then Environment chooses to execute an
action in
$\seteventsenv$. Thus instead of seeking a history in
$\ETLforest{\epsmodel}{\eventmodel}$ that reaches an objective
formula, as in epistemic planning, one seeks a \emph{strategy} for Controller: formally, it is
%We let $\Histcon$ and $\Histcon$ $\Hist{\ETLforest{\epsmodel}{\eventmodel}}$
a partial function
$\strat:\Hist{\ETLforest{\epsmodel}{\eventmodel}}(\world_\init)\rightharpoonup\seteventscon$
defined on histories of odd length (when it is the
controller's turn). An \emph{outcome} of a strategy $\strat$ is a play
$\pi=\world_\init\event_1\event_2\ldots$ in which the controller
\emph{follows} $\strat$, \ie for all $i\in\N$,
$\event_{2i+1}=\strat(\play_{\le 2i}) \in \seteventscon$, while the other
actions, of the form $\event_{2i+2}$, are selected by
the environment.
% We let $\Out{\strat}$ be the set of all outcomes of $\strat$.
A  strategy $\strat$ for Controller is \emph{winning} for an objective formula $\formula\in\languageEL$ if for every outcome $\play$ of $\strat$, there exists $i\in\N$ \suth
$\ETLforest{\epsmodel}{\eventmodel},\play_{\le i}\models
\formula$.

% \bm{est-ce qu'on a besoin d'un sous-ensemble $\multipoint$ aussi ici?}
% \bm{j'ai enlevé winning to reach, c'est lourd et pas correct. Y'a reachability
%   dans le titre (vous en pensez quoi ?) et j'ai précisé dans l'intro
%   qu'on fait que reachability}
\begin{definition}[The controller synthesis problem]
	\label{definition:epistemicgameproblem}
	~
	
	\problemdefinition{a pointed epistemic model $\kripkemodel,
          \world_\init$, action model $\eventmodel$ with
          $\setevents=\seteventscon\uplus\seteventsenv$,  and  an objective
		$\formula \in \languageEL$;}{yes if there exists a
                winning  strategy for Controller for objective $\formula$; no otherwise.}
	%The sequence $(\modelAction_1,\event_1),\ldots,(\modelAction_\param,\event_\param)$ is called a \emph{plan}. 
\end{definition}

\begin{remark}
  \label{rem-synthesis}
  Formally, we define and study the problem of existence of a
  strategy. We take the liberty to call the problem \emph{controller
    synthesis} because all the algorithms we provide can produce a
  winning strategy whenever there exists one. The same remark applies to
  the distributed strategy synthesis problem defined in the next section.
\end{remark}

As the plan existence problem reduces to the controller synthesis
problem, the undecidability of the former entails the one of the
latter.  We next establish that in all known subcases where the plan
existence problem is decidable, so is the controller synthesis
problem.

\subsection{The case of non-expanding action models}

We consider so-called \emph{non-expanding} action models where actions
do not expand epistemic models when executed, like public actions. For
this type of actions, the search space is finite and thus the problem
is decidable. % decidability of the problem is fairly easy to argue
% \fs{attention, " is fairly easy to argue" fait prétentieux et
%   n'apporte aucune information. On est limité à 6 pages. faut
%   supprimer ça et dire par ex : the search-space is finite thus the
%   problem is decidable.}, but
We establish the precise computational
complexity of the problem in these cases.

\begin{theorem}
	\label{theorem:epistemicgameproblem-publicannouncements-PSPACEcomplete}
	When actions are public announcements, the controller synthesis problem is \PSPACE-complete.
\end{theorem}

  \begin{proof}
      Since applying public
  announcements to epistemic models only removes worlds, and does not
  change those that remain, the number of
  successive public announcements to consider can be bounded by the number of
  worlds in the initial epistemic model.
We can thus solve the problem with an alternating algorithm that runs in polynomial time,
   guessing existentially actions of the controller and universally
  those of the environment. The algorithm is given in Figure~\ref{figure:algoPSPACEpublicannouncements}.

  \begin{figure}
  	\begin{algo}
  		\begin{algoblocprocedure}{controllerSynthesisPublicAnnouncements($\kripkemodel,
  				\world_\init$, $\eventmodel$,
  				$\setevents=\seteventscon\uplus\seteventsenv$, 
  				$\formula$)}

  			\vspace{1ex}
  			set $\kripkemodel,
  			\world_\init$ as the current pointed epistemic model;

  			\begin{algoblocfor}{$i := 0$ to the number of
                            worlds in $\kripkemodel$}

                          \vspace{1ex}                          
  				\algoif the current pointed epistemic model satisfies $\formula$ \algothen \algoaccept;

				\algoif $i$ is even \algothen
                                existentially choose $\event \in
                                \seteventscon$ that is executable in
                                the current pointed epistemic model
                                (fail if no such action exists);
				
				\algoif $i$ is odd \algothen
                                universally choose $\event \in
                                \seteventsenv$ that is executable in
                                the current pointed epistemic model
                                (fail if no such action exists);
				
				set $\kripkemodel \otimes \eventmodel, (\world, \event)$ as the current pointed epistemic model, where $\kripkemodel, \world$ was the previous current pointed epistemic model
  			\end{algoblocfor}

  			\algoreject
  			
  			\end{algoblocprocedure}
  	\end{algo}
\caption{Alternating algorithm for deciding in polynomial-time the controller synthesis problem when actions are public announcements.  	\label{figure:algoPSPACEpublicannouncements}}
  \end{figure}

  We conclude by recalling that alternating
  polynomial time corresponds to deterministic polynomial space~\cite{chandra1976alternation}.
Note that checking epistemic formulas (preconditions and
$\formula$) in epistemic models, and thus also computing the update product, can be performed in
  polynomial time.

	\newcommand{\timei}{\lbigand_{j=1}^{i-1} \lknow \agent \lnot q_j \land \lbigand_{j=i}^{2k}  \lknowpos \agent  q_j}
  PSPACE-hardness comes from a polynomial reduction from TQBF (True
  Quantified Boolean Formulae) which is PSPACE-complete
  \cite{sipser2006introduction}.
  A QBF formula
  $$\exists p_1 \forall p_2 \dots \exists p_{2k-1} \forall p_{2k}
  \chi(p_1, \dots, p_{2k})$$ is transformed in the following instance
  of the controller existence problem:
	\begin{itemize}
%		\item $\agtset = \set{\agent}$;% on s'en fiche de ça
		\item $\kripkemodel$ is the pointed Kripke model made
                  up of a $\set{q_i}$-world and $\set{p_i}$-world for
                  all $i \in \set{1, \dots, 2k}$ and an extra
                  $\vide$-world $w$, which is the pointed world; the epistemic relation for agent $\agent$ is universal;

                  \begin{center}
                  	\begin{tikzpicture}
                  	\draw[rounded corners=5mm] (-2, 4.5) rectangle (2, -0.5);
                  	\node [realworldarrowfromleft] {$w : \emptyset$};
                  	\node  at (-1, 3) {$\vdots$};
                  	\node  at (1, 3) {$\vdots$};
                  	\foreach \y/\i in {1/1, 2/2, 4/{2k}} 
                  	{
                  	\node at (-1, \y) {$w_{\i} : \set{p_{\i}}$};	
                  	\node at (1, \y) {$u_{\i} : \set{q_{\i}}$};}
                  	\end{tikzpicture}
                  \end{center}
		\item The possible announcements are
\[\varphi_{\neg p_i}=\timei \land \lnot p_i
                    \land  \lnot q_i \quad     \mbox{and}  \quad \varphi_{p_i}=\timei \land \lnot q_i\] for  $i \in \set{1,\dots, 2k}$.

                  They belong to the controller when $i$ is odd, and to the environment when $i$ is even;
		\item The goal is $\lbigand_{j=1}^{2k} \lknow \agent \lnot q_j \land \chi(\lknowpos \agent p_1, \dots, \lknowpos \agent p_{2k})$.
	\end{itemize}
        In the model $\modelM$, worlds $w_i$ are
        used to encode assignments of truth values to atoms $p_i$:
        removing world $w_i$ means setting $p_i$ to true, while
        keeping it means setting $p_i$ to false. Worlds $u_i$, bearing atoms $q_i$, are
        used to enforce that the value of each atom $p_i$ is set
        exactly once. In announcements $\phi_{p_i}$ and $\phi_{\neg
          p_i}$, conjunct $\timei$ implies
        that worlds $u_1,\ldots,u_{i-1}$ have already been removed,
        while worlds $u_i,\ldots,u_{2k}$ are still in the model. Thus
 announcements $\phi_{p_i}$ and $\phi_{\neg p_i}$
        are possible in round $i$, and only there.

        Now observe that announcement $\phi_{\neg p_i}$, because of
        conjunct $\neg p_i \land \neg q_i$, removes both world $w_i$
        and world $u_i$, thus setting $p_i$ to true. Announcement
        $\phi_{p_i}$ instead removes only world $u_i$, thus setting
        $p_i$ to true.

In the goal formula, $\lbigand_{j=1}^{2k} \lknow \agent \lnot q_j$
means that all the variables $p_1, \dots, p_{2k}$ have been
assigned. The clause
$\chi(\lknowpos \agent p_1, \dots, \lknowpos \agent p_{2k})$ is the
formula $\chi(p_1, \dots, p_{2k})$ in which we replaced $p_i$ by
$\lknowpos \agent p_{i}$, which holds if and only if world
$w_i$ has not been removed by announcements, \ie if and only if
announcement $\phi_{p_i}$ was chosen at round $i$.

The fact that the announcements that assign values to
$p_1, p_3, \dots$ are assigned to the controller and that the
announcements that assign values to $p_2, p_4, \dots$ are played by
the environment reflects the alternation of quantifiers in the formula
$\exists p_1 \forall p_2 \dots \exists p_{2k-1} \forall p_{2k}
\chi(p_1, \dots, p_{2k})$.
\end{proof}
%
%\bm{I think the proof for the lower bounds can be removed if we need space, and we just
%say reduction from QBF.}
%\fs{maybe}
%%sûr!

\begin{theorem}
	\label{theorem:epistemicgameproblem-publicactions-EXPTIMEcomplete}
The controller synthesis problem for public actions is \EXPTIME-complete.
\end{theorem}

\isnotlongversion{
\begin{proofsketch}
As for public announcements,  applying a public action in a model does not add
worlds. However it may change facts in the worlds, so that linear
sequences of actions may not suffice.  Nonetheless, we can turn the alternating algorithm from the proof of
Theorem~\ref{theorem:epistemicgameproblem-publicannouncements-PSPACEcomplete}
into one that runs in  polynomial-space.
  The \EXPTIME-membership of the problem follows from the fact that
alternating polynomial space corresponds to exponential time~\cite{chandra1976alternation}. 
  \EXPTIME-hardness is obtained by reduction from the conditional
  planning problem, a variant of classical planning with non-deterministic actions~\cite{DBLP:journals/jair/LittmanGM98,DBLP:conf/aips/Rintanen04a}.
  % and from a polynomial reduction of the conditional planning problem
  % to our controller synthesis problem with only public actions.
\end{proofsketch}
}

  \begin{proof}
    As for public announcements,  applying a public action in a model does not add
worlds. However it may change facts in worlds, so that 
sequences of actions of linear length may not suffice.  Nonetheless, linear space is enough to store the current pointed
   epistemic model, and we can turn the alternating algorithm from the proof of
Theorem~\ref{theorem:epistemicgameproblem-publicannouncements-PSPACEcomplete}
into one that runs in  polynomial space.  The new algorithm is given in Figure~\ref{figure:algoEXPTIMEpublicactions}, in which we do not bound the length of the sequence of actions.  Note that the algorithm may not terminate, but it is folklore that we can add a counter to ensure termination while staying in polynomial-space; we do not present these tedious technicalities here.

\begin{figure}
	\begin{algo}
		\begin{algoblocprocedure}{controllerSynthesisPublicActions($\kripkemodel,
				\world_\init$, $\eventmodel$,
				$\setevents=\seteventscon\uplus\seteventsenv$, 
				$\formula$)}

                              \vspace{1ex}
                              
			set $\kripkemodel,
			\world_\init$ as the current pointed epistemic model;
			
			$i := 0$;
			
			\begin{algoblocwhile}{the current pointed epistemic model does not satisfy $\formula$}

                          \vspace{1ex}
				\algoif $i$ is even \algothen existentially choose $\event \in \seteventscon$ that is executable in the current pointed epistemic model (fail if no such action);
				
				\algoif $i$ is odd \algothen universally choose $\event \in \seteventsenv$ that is executable in the current pointed epistemic model (fail if no such action);
				
				set $\kripkemodel \otimes \eventmodel, (\world, \event)$ as the current pointed epistemic model, where $\kripkemodel, \world$ was the previous current pointed epistemic model;

				$i := 1 - i$;
				
			\end{algoblocwhile}

			\algoaccept
			
		\end{algoblocprocedure}
	\end{algo}
	\caption{Alternating algorithm for deciding in polynomial-space the controller synthesis problem when actions are public actions.  	\label{figure:algoEXPTIMEpublicactions}}
\end{figure}

  The \EXPTIME-membership of the problem follows from the fact that
alternating polynomial space corresponds to exponential time~\cite{chandra1976alternation}. 

\EXPTIME-hardness is obtained by reduction from the conditional
  planning problem, a variant of classical planning with full observability
and  non-deterministic actions, where the plan should lead to a situation satisfying
  the goal no matter how nondeterminism is resolved. However the plan can
  depend on how nondeterminism is resolved, hence the name
  ``conditional plan''~\cite{DBLP:journals/jair/LittmanGM98,DBLP:conf/aips/Rintanen04a}.
  % The membership in \APSPACE = \EXPTIME \cite{chandra1976alternation} is
  % provided by a  alternating algorithm  that
  % guesses  actions existentially for the
  % controller and universally  for the environment. Because public
  % actions are non-expending, linear space is enough to store the current pointed
  % epistemic model. 
%\sp{alternatively, one may reduce conditional planning know to be EXPTIME-complete [rintannen, littman], this is preferred for an AI conference, and the relevance of the controller existence problem}

Stated in our terms, conditional planning essentially corresponds to a
particular case of controller synthesis which is purely boolean (no
epistemic content), but where actions chosen by the controller have
nondeterministic effects, and the environment resolves nondeterminism. Since
everything is purely boolean, the initial situation is a one-world
epistemic model, \ie a valuation over a finite set of atoms $\AP$, the
goal is a boolean formula over $\AP$, and
each action is a one-state action model with nondeterministic
postcondition.
A conditional plan is then a winning strategy for the
controller. Thus, to
finish the reduction, we only have to show how to simulate
nondeterministic actions in our setting.

\newcommand{\postvec}{ \overrightarrow{\post}}
\newcommand{\operatorprecondition}{\phi}

In~\cite{DBLP:journals/jair/LittmanGM98,DBLP:conf/aips/Rintanen04a},
a nondeterministic action is modelled as a tuple
$\tuple{\operatorprecondition,\postvec}$ where
$\operatorprecondition$ is a Boolean precondition,
and $\postvec$ is a finite set $\set{\post_1, \ldots, \post_n}$, where
$\post_i:\AP\to \languagePropositional$ is a postcondition. The
idea is that in each round the controller chooses an action among those whose
precondition is true, and
the environment  resolves the
non-determinism by choosing which postcondition of $\postvec$ to apply
to the current valuation.  % Conditional planning can be polynomially
% reduced to the controller existence problem with public actions as
% follows: the initial pointed epistemic model $\kripkemodel, \world$
% contains a single world that is the initial state in the conditional
% planning instance.
% \bm{the notion of
%    repertoire has not been introduced} %\fs{DONE}
For each nondeterministic action
$\tuple{\operatorprecondition,\set{\post_1,\ldots,\post_n}}$  of the conditional planning
instance, we create one action model for the controller that stores in a
finite-domain variable $action$ which action has been played, and $n$
actions for the environment that correspond to the different
possible poscondtions.
 The action for the controller is defined as follows: 
  \[\eventinfigure{\operatorprecondition}{action := \tuple{\operatorprecondition,\postvec}}\]  
 %    That action in the action model assigns the operator/action to the .
 % The environment repertoire aims at resolving the
 %  non-determinism. For all operator/actions $\tuple{\operatorprecondition,\postvec}$ where $\postvec = \set{\post_1, \dots, \post_n}$, for all $i \in \set{1, \dots, n}$, we add the following action that the environment can play:
  
  while the actions for the environment are, for each $i\in\set{1,\ldots,n}$,
  \[\eventinfigure{action =
      \tuple{\operatorprecondition,\postvec}}{\post_i}\]
  
  % That action can be played by the environment if the operator/action chosen by the controller is  $\tuple{\operatorprecondition,\postvec}$. Its postcondition is exactly the assignment indicated by $\post_i$.
  % %
  % The goal is the same as in the conditional
  % planning instance.
  This finishes the proof, and also shows how the controller synthesis
  problem subsumes conditional planning.
  We now present  an alternative proof that reduces
   from a more basic decision problem called $G_4$,
  introduced by Chandra and Stockmeyer
  \cite{DBLP:journals/siamcomp/StockmeyerC79}.
This is essentially an adaptation of the 
proof from~\cite{DBLP:journals/jair/LittmanGM98} for the
EXPTIME-hardness of conditional planning.
  
  \newcommand{\propplayerexists}[1]{p_{#1}}
  
  \newcommand{\propplayerforall}[1]{q_{#1}} The input to the $G_4$
  problem is a 13-DNF formula over  $2k$ atomic
propositions
  $\propplayerexists1, \dots, \propplayerexists k$, $\propplayerforall1,
  \dots, \propplayerforall k$ and an initial valuation. Atoms
  $\propplayerexists1, \dots, \propplayerexists k$ are controlled by
  the controller (the existential player) while 
  $\propplayerforall1, \dots, \propplayerforall k$ are controlled by
  the environment (the universal player).  Now, the following game is
  played: each player, when it is her turn to play, flips the
  assignment of one of the variables she controls, and turns alternate. The game stops when the
  13-DNF formula becomes true, and the winner is the player that made  the
last move. % The decision problem $G_4$ takes as input a
  % 13-DNF formula, the two sets of variables, a valuation.
  An instance of the $G_4$ problem
  is positive if the controller has a winning strategy. 

We construct the following instance of our controller synthesis problem. The initial epistemic model is made up of one world, whose valuation is the initial valuation of $G_4$. Actions of the controller are:

\begin{center}
	$\eventinfigure{\top}{\propplayerexists{1} := \lnot \propplayerexists{1} } ~~~~\dots~~~~
	\eventinfigure{\top}{\propplayerexists{k} := \lnot \propplayerexists{k} }
	 $
\end{center}

Actions of the environment are:

\begin{center}
	$\eventinfigure{\top}{\propplayerforall{1} := \lnot \propplayerforall{1} } ~~~~\dots~~~~
	\eventinfigure{\top}{\propplayerforall{k} := \lnot \propplayerforall{k} }
	$
\end{center}

The goal is the 13-DNF formula.
%EXPTIME-hardness comes from a polynomial reduction of the TWO PLAYER CORRIDOR TILING \cite{DBLP:journals/jcss/Chlebus86} which is EXPTIME-complete. In that problem, a rectangle of width $n$ has to be filled with four-colored tiles of the form \tikz[scale=0.3]{\draw (0, 0) rectangle (1, 1); \draw (0, 0) -- (1,1); \draw (0, 1) -- (1, 0);} from bottom to top, left to right, the possible tiles is given in input. The bottom line of tiles is given in the input. The two players alternate and the first player wins when a particular tile is set. The initial epistemic model contains $n$ worlds: the $i$-th world is labelled by propositions $\set{x_1, \dots, x_\ell}$ that encodes the binary representation of $i$ and a proposition $t_i$ that says that there is a tile $t_i$ at column $i=1..n$. Moreover, in all worlds, propositions $\set{x'_1, \dots, x'_\ell}$ are false (these variables represent an integer that says which is the current column to be filled in). Sets $\seteventscon$ and $\seteventsenv$ are (copies of) the set of public actions that encode the tile at column $i$ and the increment of $\set{x'_1, \dots, x'_\ell}$. %C'est technique, pas la place de détailler plus
\end{proof}

Theorem~\ref{theorem:epistemicgameproblem-publicactions-EXPTIMEcomplete}
also generalises to other non-expanding actions models such as the
so-called \emph{separable action models} \cite{PhDTristanCharrier},
where the preconditions of any two actions in the same connected
component are logically inconsistent.

\subsection{The case of propositional action models}
\label{sec-ctr-prop}

% Several approaches have been used to prove decidability of the plan % (see Section~\ref{sec-EL})
% existence problem. % has been proved in several works. %  is
% decidable.
% ~\cite{DBLP:conf/ijcai/YuWL13,DBLP:journals/corr/AucherMP14,AiML2018Gaetanetal}.
To solve our controller synthesis problem we rely on the approach followed
in~\cite{DBLP:phd/hal/Maubert14} to solve the plan existence problem
for propositional actions. This approach has two main ingredients: (I1)
when $\eventmodel$ is propositional, the generated structure
$\ETLforest{\epsmodel}{\eventmodel}$ can be represented finitely, and
(I2) one can decide the existence of a winning strategy in a certain
class of two-player games with epistemic objectives.

%However, while epistemic planning only yields a degenerate situation where the environment is idle, we here use the full power of the latter result, and establish the following.

\begin{theorem}
	\label{theorem:epistemicgameproblem-prepostBoolean-decidable}
	When action models are propositional, the controller synthesis
        problem is decidable, and in \kEXPTIME[{(k+1)}] if the
        objective's modal depth is bounded by $k$.
\end{theorem}

%
%\todo{define modal depth?}
%%done

We devote the rest of this section to prove
Theorem~\ref{theorem:epistemicgameproblem-prepostBoolean-decidable}, which requires to introduce particular game arenas. %  which is more general than what we
 % need for controller synthesis, but that we will also use in the next
 % section on distributed strategy synthesis.

 %\bm{donner directement définition générale multiplayer game arenas?}
%\fs{pourquoi faire?}
\begin{definition}%[Two-player epistemic game arena]
  \label{def-etl-model}
  A \emph{two-player epistemic  game arena} is a structure\linebreak[4]
  $\etlmodel=(\setworlds,\world_\init,\trans,(\epistemicrelation{\agent})_{\agent
    \in \agtset},\valworlds)$ where $(\setworlds,(\epistemicrelation{\agent})_{\agent
    \in \agtset},\valworlds)$ is an epistemic model, 
  $\setworlds=\setworlds_0\uplus\setworlds_1$ is
  partitioned into the positions of players 0 and 1, $\world_\init$ is
  an \emph{initial world} and $\trans\subseteq\setworlds\times\setworlds$ is a
  \emph{transition relation}.
\end{definition}

% Observe that players and agents may be different, and they have
% different roles: players choose how the system evolves, while agents
% observe the system. Typically in epistemic planning, one external player
% chooses the events while the agents involved in the DEL models only
% observe. However, in case $\setplayers=\agtset$ and each
% $\epistemicrelation\agent$ is an equivalence relation, these
% structures correspond to multi-player turn-based arenas with imperfect
% information from~\cite{peterson2002decision}. 

A \emph{play} in a game arena $\etlmodel$ is an infinite sequence of worlds
$\play=\world_0\world_1\world_2\ldots$ such that for all $i\in\N$,
$\world_i\trans\world_{i+1}$, and a \emph{history} is a finite
nonempty prefix
of a play. We let $\Play{\etlmodel}$ and $\Hist{\etlmodel}$ be the
sets of plays and histories in $\etlmodel$,
respectively. Accessibility relations
$(\epistemicrelation{\agent})_{\agent \in \agtset}$ are extended to
histories to interpret epistemic formulas: two histories
$\hist=\world_0\ldots\world_n$ and $\hist'=\world'_0\ldots\world'_m$
are related by $\epistemicrelation{\agent}$ whenever $n=m$ (same
length) and $\world_i\epistemicrelation{\agent}\world'_i$ for every
$i\le n$.

A \emph{strategy} for player $0$ is a partial function
$\strat:\Hist{\etlmodel}\rightharpoonup\setworlds$ such that for every $\hist$
with  $\last(\hist)\in\setworlds_0$: $\last(\hist)\trans\strat(\hist)$.
A play
$\play=\world_\init\world_1\world_2\ldots$ is an \emph{outcome} of $\strat$ if for every $i\in\N$
with $\play_i\in\setworlds_1$, we have
$\play_{i+1}=\strat(\play_{\le i})$.
Strategy $\strat$ is
\emph{winning} for an epistemic objective $\varphi \in \languageEL$, if for every outcome $\play$ of $\strat$ 
there is some $i\in\N$ with $\play_{\le i}\models \phi$.
%The following result is proved in~\cite{DBLP:journals/iandc/BozzelliMP15}:
\begin{theorem}[\cite{DBLP:journals/iandc/BozzelliMP15}]
  \label{theo-unif-strat}
  The existence of a winning strategy for player 0 in an epistemic game $\etlmodel$
  for an epistemic objective $\formula$ of
  modal depth $k$ can be
  decided in time $k$-exponential in $|\etlmodel|$ and $|\formula|$. %the size of $\etlmodel$.
\end{theorem}

%We apply this result to our controller synthesis problem.
 We show that the controller synthesis problem for propositional action models
 reduces to solving an epistemic game:
 
\begin{proposition} % Proposition 27, sec 7.2.3
  \label{prop-regular}
% One can effectively
  % transform an
Given an instance  $((\epsmodel,\world_\init),\eventmodel,\formula)$ 
of the controller synthesis problem  
where $\eventmodel$ is propositional, one can construct
  % of the former into an instance of the latter
a game arena $\etlmodel$ 
such that Controller wins in
$((\epsmodel,\world_\init),\eventmodel,\formula)$ iff Player 0 wins in
$\etlmodel$ for objective $\formula$ and 
$|\etlmodel|\le |\epsmodel|+|\eventmodel|\times 2^{m+1}$, where $m$ is
the number of atomic propositions involved in $\epsmodel,\eventmodel$
and $\formula$.
% and such
  % that winning controller strategies for $\formula$ and winning strategies of player 1 for $\formula$ in $\etlmodel$ correspond.
\end{proposition}

\begin{proof}
Let $\epsmodel=(\setworlds,(\epistemicrelation{\agent})_{\agent
    \in \agtset},\valworlds)$  and 
  $\eventmodel=(\setevents,\{\releventsi\}_{\agent\in\agents},\pre,\post)$, and
    let $\atmsetf$ be the atomic propositions involved. The worlds of the
    game arena $\etlmodel$ that we build are either worlds
    $w\in \epsmodel$ or tuples $(\event,\val,i)$ where
    $\event\in\eventmodel$ represents the last action performed,
    $\val\in 2^{\atmsetf}$ is the current valutation, and $i\in\{0,1\}$
    indicates whose turn it is to play: 0 for Controller and 1 for
    Environment.  In an initial world $\world$, Controller can choose
    an action $\event\in\seteventscon$ such that
    $\world\models\pre(\event)$ and move to
    $(\event,\postval(\world,\event),1)$; in a world of the form
    $(\event,\val,i)$, if $i=1$ (resp., $i=0$), Environment (resp,
    Controller) chooses an action $\event'\in\seteventsenv$ (resp.,
    $\event'\in\seteventscon$) such that $\val\models\pre(\event')$,
    and moves to $(\event',\postval(\world,\event),1-i)$.

Formally, 
  % a propositional
  % action model with $\setevents=\seteventscon\uplus\seteventsenv$,
  letting $\setevents_0=\seteventscon$ and $\setevents_1=\seteventsenv$,
we  define the epistemic game arena $\etlmodel=(\setworlds',\world_\init,\trans,(\epistemicrelation{\agent}')_{\agent
    \in \agtset},\valworlds')$ as follows:
Let     $\setworlds_0=\setworlds\union \setevents\times2^\AP\times\{0\}$,
 $\setworlds_1=\setevents\times 2^\AP\times\{1\}$, and $\setworlds'=\setworlds_0\cup\setworlds_1$.
For transitions, for $w\in\setworlds$, $\event,\event'\in\setevents$ and
$\val,\val'\in 2^\AP$, we let $\world \trans (\event,\val,1)$ if
    $\event\in\seteventscon, \world \models \pre(\event)$
 and $\val=\postval(\world,\event)$,
and $(\event,\val,i)\trans(\event',\val',1-i)$ if 
  $\event'\in\setevents_i, \val \models \pre(\event')$
 and $\val'=\postval(\val,\event)$.
We also let $\world\epistemicrelation{\agent}'\world'$ if
$\world\epistemicrelation{\agent}'\world'$, and 
  $(\event,\val,i)\epistemicrelation'(\event',\val',i)$ 
if  $\event\releventsi\event'$, and finally
 $\valworlds'(\world)=\valworlds(\world)$ and $\valworlds'(\event,\val,i)=\val$.

 The structure given by the set of histories
$\Hist{\etlmodel}$ and relations $\epistemicrelation{\agent}'$
extended to these histories is isomorphic to
$\ETLforest{\epsmodel}{\eventmodel}$, and histories of odd (resp. even)
length  in $\ETLforest{\epsmodel}{\eventmodel}$ correspond to histories
that end in $\setworlds_0$ (resp. $\setworlds_1$) in $\etlmodel$
(provided they start in $\world_\init$). It follows that there is a winning
strategy for Controller in the original problem if and only if
there is winning strategy for Player 0 in $\etlmodel$ with objective
$\formula$.
\end{proof}

Note that, as stated in Proposition~\ref{prop-regular}, the resulting game arena is indeed polynomial in the size of the epistemic and action models, but
it is exponential in the number of atomic propositions involved in the
problem. This is because states of the
game arena contain all possible valuations.
Theorem~\ref{theorem:epistemicgameproblem-prepostBoolean-decidable}
now follows from Theorem~\ref{theo-unif-strat} and
Proposition~\ref{prop-regular}.

With the controller synthesis problem we enriched  epistemic planning with an
adversarial environment.  Still, as in epistemic planning, the agents 
 are mere observers. We now make a step
further and make the agents players of the game.

%%% Local Variables:
%%% mode: latex
%%% TeX-master: "main"
%%% End:

\section{Distributed strategy synthesis}
%===================================
\label{section:imperfectinfoepistemicgames}
In this section agents are no more passive, but instead they are players who choose themselves the
actions that occur. 
The set $\agtset$ of agents is split into two teams $\agtsetexists$ and
$\agtsetforall$ that play against each other, and we may say \emph{players} instead of agents.

 \subsection{Setting up the game}
 Unlike the external controller from the previous section, our players
 now have imperfect information. The fundamental feature of
 games with imperfect information is that when a player cannot
 distinguish between two different situations, a strategy for this
 player should prescribe the same action in both situations. All
 the additional complexity in solving games with imperfect information
 compared to the perfect information setting arises from this
 constraint. Such strategies are often called \emph{uniform} or
 \emph{observation-based} (see for instance~\cite{DBLP:journals/jcss/Reif84,van2001games,apt2011lectures}). % ,herzig2006knowing,DBLP:journals/iandc/BerwangerCWDH10,DBLP:journals/logcom/PileckiBJ17}).
 % Additionally, it is
Since games with imperfect information consider S5
 models, \ie where accessibility relations are equivalence relations,
and it is unclear what  uniform strategies
mean in non-S5 models\footnote{Actually the usual definition seems to make sense for KD45, i.e. models whose relations are serial, transitive, and Euclidean. But
	this would be highly non-standard, and since all the literature on 
	games with imperfect information considers S5 models.},  we also assume from now on that all epistemic
and action models are S5. We stress this assumption by writing
 $\epistemicrelationequiv$ (resp.\ $\epistemicrelationeventsequiv$)
 instead of $\epistemicrelation\agent$ (resp.\
 $\epistemicrelationevents\agent$).

%Consider an initial pointed S5 Kripke model $\kripkemodel,\world$, and an S5 event model $\eventmodel$.
%We let $\agentofevent{\event}$ be the player that owns action $\event$.\bm{utilisé?}
We start from an initial pointed epistemic model
$\kripkemodel,\world$, and an action model
$\eventmodel$ whose set of actions % , we may call \emph{actions},
is
partitioned into subsets $(\seteventsfor\agent)_{\agent \in \agtset}$ of
actions for each player. 
The game we describe is turn-based. We use 
the variable $\variableturn$ ranging over $\agtset$ to
represent whose turn it is to play. 
% Given a valuation $\val\in 2^\AP$, we let $\agent(\val)$ be the unique
% agent $\agent$ such that $\val\models\is\variableturn\agent$, and for
% a world $\world$ we
% define $\agent(\world)$ similarly.
We require that for each $\agent\in\seteventsfor\agent$, %$\event\in\seteventsfor\agent$,
$\pre(\event)$ implies $\is\variableturn\agent$, 
and that postconditions for variable $\variableturn$ do not
    depend on the current world, but instead the next value of
    $\variableturn$ is completely determined by the
    action only.

Moreover, in order to obtain a proper imperfect-information game, we demand the following
hypotheses: %\fs{il faudrait justifier ces hypothèses dans les démonstrations}

\paragraph*{Hypotheses on $\kripkemodel$ and $\eventmodel$}
  \label{hyptothesis:conditionsinputstrategyexistenceproblem}
\begin{enumerate}
\item[(H1)] %\label{hyp-start}
  {\it\bf The starting player is known:}
  there is a player $\agent$ such that for all
  $\worldb \in \setworlds$,
  $\kripkemodel,\worldb \models \is\variableturn\agent$;
% \item {\it } For all events $\event$, there exists a unique agent $\agent$,  noted $\agentofevent{\event}$, such that the precondition of $\event$ implies $\variableturn = \agent$;
% \item {\it }for all $\event\in\seteventsfor\agent$, $\agentofevent{\event}=\agent$
% \item\label{hyp-myevents} {\it Players know what they
%   play:} for all events $\event, \event' \in
%   \seteventsfor\agent$, $\event R_{\agent} \event'$
%   implies $\event = \event'$;\bm{do we need this?}
% \item \label{hyp-turn-based} {\it The game is turn
%   based:}
%   for each event $\event$, the postconditions
%   $\post(\event)(\variableturn_\agent)$ for $\agent\in\agtset$
%   are mutually exclusive, and their disjunction is a tautology;
\item[(H2)] %\label{hyp-turn}
  {\it\bf The turn stays known:} for all
  actions $\event, \event' $ and agent $\agent$, if
  $\event R_{\agent} \event'$, then $\event$ and
  $\event'$ assign the same value to
  $\variableturn$.
\item[(H3)] %\label{hyp-rep}
  {\it\bf Players know their available actions:} if agent
  $\agent$ can execute $\event$ after history
  $\hist$, then she can also execute it after every history
  $\hist'$ with 
  $\hist\epistemicrelationequiv[\agent]\hist'$. %in $\ETLforest{\epsmodel}{\eventmodel}$.
\end{enumerate}

All these hypotheses can be either enforced syntactically or checked  in the different decidable cases we
consider in the rest of this work (see the long version for detail).

We now define formally the notion of uniform strategies. 
             
\begin{definition}[Uniform strategy]
  A strategy $\strat$ for player $\agent$ is \emph{uniform} if for every pair
  of histories $\hist,\hist'$ where it is player $\agent$'s turn, $\hist \epistemicrelationequiv \hist'$ entails 
  $\strat(\hist)=\strat(\hist')$.
\end{definition}

In the rest of this section, a strategy of a player in $\agtsetexists$
is implicitly uniform. When one selects a strategy for each player in
$\agtsetexists$, the result is called a \emph{distributed strategy}, and an
\emph{outcome} of a distributed strategy is a play in which all players in
$\agtsetexists$ follow their prescribed strategy. A distributed strategy is
\emph{winning} for an objective formula $\formula$ if all its outcomes eventually satisfy
$\formula$.

\subsection{The distributed strategy synthesis problem}

We study the existence of a distributed strategy for players in
$\agtsetexists$ that ensures to reach an epistemic goal property.

%We consider the following decision problem on DEL games.

\begin{definition}[Distributed strategy existence problem]
	\label{definition:multiepistemicgameproblem}
	~
	
	\problemdefinition{a pointed epistemic model $\kripkemodel,
          \world$ and an action model $\eventmodel$
partitioned into
          $(\seteventsfor\agent)_{\agent \in \agtset}$ that satisfy
          hypotheses~(H1)-(H3), and an objective formula 
          $\formula \in \languageEL$;}{yes if there exists a winning distributed
          strategy for players in $\agtsetexists$; no otherwise.}
	%The sequence $(\modelAction_1,\event_1),\ldots,(\modelAction_\param,\event_\param)$ is called a \emph{plan}. 
\end{definition}

%We first  transfer a classic result on games with imperfect information to our setting:
Unlike the controller synthesis problem which we proved 
decidable for propositional actions, synthesising 
distributed strategies is undecidable for propositional actions, already for a team of two players.

%\bm{I would like to put ``strategy synthesis'' everywhere}
%\fs{no because we study a decision problem. Strategy synthesis would be the }

%In this subsection, we define the epistemic team planning that extends epistemic planning for a single agent and that respects necessary assumptions on imperfect information games.

\subsection{Undecidability for two existential players}
%--------------------------------------------------------
\newcommand{\teamdfagame}{\textsf{TEAM DFA GAME}\xspace}

The following Theorem~\ref{theorem:uniformstrategyundecidable} is a reformulation in our setting of the classical undecidability
result from Reif and Peterson \cite{DBLP:conf/focs/PetersonR79}. However, we decide to promote an existing elegant reformulation of that very same result, called \teamdfagame~\cite[Def.~1,
p.~14:7]{DBLP:conf/fun/CoulombeL18}, that can be reduced to our
distributed strategy synthesis problem. % the following.
\begin{theorem}
	\label{theorem:uniformstrategyundecidable}
	The distributed strategy synthesis problem is undecidable,
        already for a propositional action model and two existential
        players against one universal player. %  if the existential team has two players, the universal
        % team has one player, and the action model is propositional.
\end{theorem}

\begin{proof}
	\newcommand{\teamdfagameplayera}{a}
	\newcommand{\teamdfagameplayerb}{b}
	\newcommand{\insansespace}{{\in}}
	\newcommand{\notinsansespace}{{\not\in}}
	 \newcommand\variablestep{\fontforfinitedomainvariables{stp}}
	\newcommand\variablegameover{\fontforfinitedomainvariables{lost}}
	\newcommand\descriptioncooleventoneline[2]{{\pre: \ensuremath{#1}; 
		\post: \ensuremath{#2}}}

	\newcommand\descriptioncoolevent[2]{{\pre: \ensuremath{#1}; 
		
		\post: \ensuremath{#2}}}
	\renewcommand{\playerexists}[1]{#1}
        	\newcommand{\bitone}{\beta}
	\newcommand{\bittwo}{\beta'}
	\newcommand{\bitanswerone}{m}
	\newcommand{\bitanswertwo}{m'}

        The proof is given by reduction from the problem
        \teamdfagame~\cite[Def.~1, p.~14:7]{DBLP:conf/fun/CoulombeL18}, shown to
        be undecidable.

        	 We consider a two-versus-one (players
        $\teamdfagameplayera$ and $\teamdfagameplayerb$ versus player
        $\forall$) team game played on a deterministic finite
        automaton (DFA) $\automaton$ whose alphabet is $\set{0, 1}$,
        whose set of states is $Q$, initial state is $q_0$, transition
        function $\delta$. Special subsets of states $F_\exists$ and
        $F_\forall$ are given. The game starts with $\automaton$ being
        in state $q_0$. Each round is divided in six steps:

        \begin{enumerate}

	\item if the current state $q$ is in $F_\exists$ then team $\set{\teamdfagameplayera, \teamdfagameplayerb}$ wins; if the current state $q$ is in $F_\forall$ then team $\set{\forall}$ wins;
	\item Player $\forall$ inputs two bits $\bitone, \bittwo$ into $\automaton$;
	\item Player $\teamdfagameplayera$ learns $\bitone$;
          \item  Player  $\teamdfagameplayera$ inputs one bit $\bitanswerone$ into $\automaton$;
	\item Player $\teamdfagameplayerb$ learns $\bittwo$;
        \item  Player $\teamdfagameplayerb$ inputs one bit $\bitanswertwo$ into $\automaton$.
\end{enumerate}
At each step, player $\forall$ has perfect information.
 TEAM DFA GAME is the decision problem: given an DFA $\automaton$, subsets of states $F_\exists$, $F_\forall$, does the team $\{a,b\}$ have a winning strategy?
 
The rest of the proof consists in representing the initial situation, the game rules and the goal of a \teamdfagame instance as a distributed strategy existence problem instance.

 \textbf{Definition of the reduction. } Let $(\automaton, F_\exists, F_\forall)$ be an instance of \teamdfagame.  Teams are $\agtsetexists = \set{\teamdfagameplayera, \teamdfagameplayerb}$ and $\agtsetforall = \set{\forall}$. We introduce a finite-domain variable $q$ that ranges over the set of states of $\automaton$.
 The variable $q$ can be represented by a finite set of atomic propositions: for example, for an automaton with 8 states from \set{0,\dots,7}, three atomic propositions, $bit_1(q)$, $bit_2(q)$ and $bit_3(q)$ so that say $(q=5)$ is the Boolean formula $bit_1(q) \land \neg bit_2(q) \land bit_3(q)$.
We prefer to keep with a finite-domain variable $q$ to avoid technicalities.
 
  We also introduce a finite-domain variable $\variablestep$ that ranges over $\set{1, 2, 3, 4, 5, 6}$. The Boolean variable $\variablegameover$ is true if the team $\agtsetexists$ has lost. We define $\kripkemodel, \world$ to be the single-world S5 epistemic model in which $\is\variableturn \forall$, $\is q {q_0}$, $\is\variablestep 1$, $\lnot \variablegameover$.
 The actions in $\seteventsfor{\forall}$ form an $\teamdfagameplayera$- and $\teamdfagameplayerb$-indistinguishably equivalence class and are of the form:
\newcommand{\serrerlesitems}{\setlength\itemsep{-1mm}}
%\fs{moins de marge pour les itemize et itemsep plus petit}
 \begin{itemize}%[leftmargin=*]
% 	\serrerlesitems
 	\item  \descriptioncooleventoneline{\is\variableturn\forall \land \is\variablestep1 \land q \insansespace F_\forall}{\variablegameover \assignment \top, \variablestep \assignment 2}
	\item \descriptioncooleventoneline{\is\variableturn\forall \land \is\variablestep1 \land q \notinsansespace F_\forall}{\variablestep \assignment 2}

\newcommand\descriptionplayeralphainputstwobits[2]{ \descriptioncooleventoneline{\is\variableturn\forall {\land} \is\variablestep2}{\bitone \assignment #1,  \bittwo \assignment #2, q \assignment \delta(q, #1#2), \variableturn \assignment \teamdfagameplayera,  \variablestep \assignment {3}}
}
\item \descriptionplayeralphainputstwobits00
\item \descriptionplayeralphainputstwobits10
\item \descriptionplayeralphainputstwobits01
\item \descriptionplayeralphainputstwobits11

 \end{itemize}

\noindent
$\seteventsfor{\teamdfagameplayera}$ is a $\teamdfagameplayerb$-indistinguishably equivalence class and contains:
 \begin{itemize}%[leftmargin=*]
% 	\serrerlesitems
	\item \descriptioncooleventoneline{\is\variableturn{\teamdfagameplayera} \land \is\variablestep3 \land \bitone}{\variablestep := 4}
	\item \descriptioncooleventoneline{\is\variableturn{\teamdfagameplayera} \land \is\variablestep3 \land \lnot \bitone}{\variablestep := 4}
	\item \descriptioncooleventoneline{\is\variableturn{\teamdfagameplayera} \land \is\variablestep4}{\bitanswerone \assignment \bottom; \variablestep := 5; q := \delta(q, 0); \variableturn := \teamdfagameplayerb}
	\item \descriptioncooleventoneline{\is\variableturn{\teamdfagameplayera} \land \is\variablestep4}{\bitanswerone \assignment \top; \variablestep \assignment 5; q \assignment \delta(q, 1); \variableturn := \teamdfagameplayerb}
\end{itemize}

\noindent $\seteventsfor{\teamdfagameplayerb}$ is an $\teamdfagameplayera$-indistinguishably equivalence class and contains:
 \begin{itemize}%[leftmargin=*]
% 	\serrerlesitems
 	\item \descriptioncooleventoneline{\is\variableturn{\teamdfagameplayerb} \land \is\variablestep5 \land \bittwo}{\variablestep \assignment 6}
 	\item \descriptioncooleventoneline{\is\variableturn{\teamdfagameplayerb} \land \is\variablestep5 \land \lnot \bittwo}{\variablestep \assignment 6}
 	\item \descriptioncooleventoneline{\is\variableturn{\teamdfagameplayerb} \land \is\variablestep6}{\bitanswertwo \assignment \bottom; \variablestep \assignment 1; q \assignment \delta(q, 0); \variableturn \assignment \forall}
 	\item \descriptioncooleventoneline{\is\variableturn{\teamdfagameplayerb} \land \is\variablestep4}{\bitanswertwo \assignment \top; \variablestep \assignment 1; q \assignment \delta(q, 1); \variableturn \assignment \forall}
 \end{itemize}

 The assignments $q:=\delta(q,0)$ and $ q:=\delta(q,1)$ are shortcuts
 for some ``if statements'' on states, e.g. `if $q=5$, then $q :=
 2$'. For instance, assuming that we have eight states
 $\{s_0,\ldots,s_7\}$ which are thus representable with three bits,
 the assignment $q:=\delta(q,0)$ is simulated by the following set of
 propositional assignments: $\set{bit_i(q):=\psi_i}_{i=1..3}$,
 where $\psi_i$ is the Boolean formula
 $\bigor_{k \in {0..7} \text{ s.t. the i-th bit of $\delta(s_k,0)$ is
     1}} (q=s_k)$.

The goal formula $\formula$ is $\lnot \variablegameover \land \is\variablestep1 \land (q {\in} F_\exists)$.
\end{proof}

We now turn to decidable cases: games with imperfect information and
epistemic objectives are known to be decidable either when actions are
public \cite{DBLP:conf/atal/BelardinelliLMR17}, or when information is
hierarchical \cite{DBLP:conf/kr/MaubertM18}. We establish similar results
in our setting.

\subsection{The case of non-expanding action models}
\label{sec-non-expanding}
% -----------
% \sp{inflating or stretching, widening, expanding? to express that the updates' size remains less than or equal to the one of the initial epistemic structure. + Explain the concept.}

Theorems~\ref{theorem:epistemicgameproblem-publicannouncements-PSPACEcomplete}
and~\ref{theorem:epistemicgameproblem-publicactions-EXPTIMEcomplete}
of Section~\ref{section:controller} generalise to the distributed
strategy synthesis problem. First, we inherit the lower bounds  by
letting %the controller be the single agent of team
$\agtsetexists=\{\text{Controller}\}$ and
% the environment be the single agent of team
$\agtsetforall=\{\text{Environment}\}$, and by
making them alternate turns. Second, the upper bounds are obtained by
adapting the alternating algorithms for the upper bounds of Theorems
~\ref{theorem:epistemicgameproblem-publicannouncements-PSPACEcomplete}
and~\ref{theorem:epistemicgameproblem-publicactions-EXPTIMEcomplete}. We
need to ensure that existential choices of actions of an agent
$\agent \in \agtsetexists$ lead to a \emph{uniform} strategy. To do
that, every time agent $\agent$ picks an action $\event$, we perform an
extra universal choice over
$\epistemicrelationequiv$-indistinguishable worlds, and continue executing the algorithm from
these worlds. %\fs{tentative d'amélioration}% \bm{Pour que ça marche je
%$\agtsetforall$, and by alternating their actions. 
%Second, the upper bounds are obtained by
%adding universal choices of $\epistemicrelation\agent$-successors in
%alternating algorithms in order to ensure that the existential choices
%correspond to a \emph{uniform} strategy.
%\bm{Pour que ça marche je
  % crois qu'il faut que les joueurs existentiels connaissent les
  % actions les uns des autres, sinon on demande aux actions de marcher
  % sur des mondes qui ne sont pas atteignables par la stratégie qu'on
  % construit. Ça marche parce que les actions sont publics. Il faut le dire.} \fs{on en discutera de vive voix !}

\begin{theorem}
	\label{theorem:uniformstrategyexistence-publicannouncements-PSPACEcomplete}
	For public announcements, the distributed strategy synthesis problem is PSPACE-complete.
\end{theorem}

\begin{proof}
	The alternating algorithm that runs in polynomial time is
        similar to the one given in the proof
        of~Theorem~\ref{theorem:epistemicgameproblem-publicannouncements-PSPACEcomplete},
        except that we also add universal choices of $\epistemicrelationevents\agent$-successors for player $\agent$ in $\agtsetexists$. More precisely, the steps of the algorithm works as follows. When it is the turn of a player $\agent$ in $\agtsetforall$ to play, simply universally guess an executable action in $\seteventsfor\agent$. When it is the turn of a player $\agent$ in $\agtsetexists$ to play, then perform the following steps:
        \begin{itemize}
        	\item first existentially guess an executable action $\event$ in $\seteventsfor\agent$;
        	\item second, universally guess a $\epistemicrelation\agent$-successor of the pointed world in the current epistemic model and make it as the new pointed world;
        	\item compute the new epistemic model by executing the $\event$.
        \end{itemize} As for Theorem \ref{theorem:epistemicgameproblem-publicannouncements-PSPACEcomplete}, the length of such a sequence is bounded by the number of worlds in the initial epistemic model $\modelM$. 
    The obtained algorithm is given in Figure~\ref{figure:algodistrPSPACEpublicannouncements}.

    \begin{figure}[th]
    	\begin{algo}
    		\begin{algoblocprocedure}{DistrStratSynthesisPublicAnnouncements($\kripkemodel,
    				\world_\init$, $\eventmodel$,
    				$(\seteventsfor\agent)_{\agent \in \agtset}$, 
    				$\formula$)}

                              \vspace{1ex}
    			set $\kripkemodel,
    			\world_\init$ as the current pointed epistemic model;

    			\begin{algoblocfor}{$i := 0$ to the number of
    					worlds in $\kripkemodel$}
    				
    				\vspace{1ex}
    				\algoif the current pointed epistemic model satisfies $\formula$ \algothen \algoaccept;

    				let $\agent$ be the agent such that $\is\variableturn\agent$ is true in the current pointed epistemic model;

    				\algoif $\agent \in \agtsetexists$ \algothen
    				
    				\begin{algobloc}
                                  \vspace{1ex}
    					existentially choose $\event \in
    					\seteventsfor\agent$ that is executable in
    					the current pointed epistemic model
    					(fail if no such action exists);

    					universally choose an $\epistemicrelation\agent$-successor of the pointed world in the current epistemic model and make it as the new pointed world;
    				\end{algobloc}

                                \vspace{1ex}
    				\algoelse \algoif $\agent \in \agtsetforall$ \algothen
    	
                                \begin{algobloc}
                                  \vspace{1ex}
    				universally choose $\event \in
    				\seteventsfor\agent$ that is executable in
    				the current pointed epistemic model
    				(fail if no such action exists);
                              \end{algobloc}
                              %\vspace{1ex}
                              
    				set $\kripkemodel \otimes \eventmodel, (\world, \event)$ as the current pointed epistemic model,  where $\kripkemodel, \world$ was the previous current pointed epistemic model;
    			\end{algoblocfor}

    			\algoreject
    			
    		\end{algoblocprocedure}
    	\end{algo}
    	\caption{Alternating algorithm for deciding in polynomial-time the distributed strategy synthesis problem when actions are public announcements.  	\label{figure:algodistrPSPACEpublicannouncements}}
    \end{figure}

     Hardness follows from Theorem \ref{theorem:epistemicgameproblem-publicannouncements-PSPACEcomplete} (simply consider the controller as the single agent in $\agtsetexists$, the environment as the single agent in $\agtsetforall$, and make them alternate).
\end{proof}

\begin{theorem}
	\label{theorem:uniformstrategyexistence-publicactions-EXPTIMEcomplete}
For  public actions, the distributed strategy synthesis problem  is EXPTIME-complete.
\end{theorem}

\islongversion{
\begin{proof}
  The alternating algorithm that runs in polynomial space is similar
  to the one given in the proof of Theorem~\ref{theorem:epistemicgameproblem-publicannouncements-PSPACEcomplete}.
  The algorithm is in polynomial space for the same reason said in the proof of~Theorem~\ref{theorem:epistemicgameproblem-publicactions-EXPTIMEcomplete}. The obtained algorithm is given in Figure~\ref{figure:algodistrEXPTIMEpublicactions}.

  \begin{figure}[ht]
  	\begin{algo}
  		\begin{algoblocprocedure}{DistrStratSynthesisPublicActions($\kripkemodel,
  				\world_\init$, $\eventmodel$,
  				$(\seteventsfor\agent)_{\agent \in \agtset}$, 
  				$\formula$)}
  			
\vspace{1ex}  			
  			set $\kripkemodel,
  			\world_\init$ as the current pointed epistemic model;

  			\begin{algoblocwhile}{the current pointed epistemic model does not satisfy $\formula$}
  				
  				\vspace{1ex}

  				let $\agent$ be the agent such that $\is\variableturn\agent$ is true in the current pointed epistemic model;

  				\algoif $\agent \in \agtsetexists$ \algothen
  				
  				\begin{algobloc}
  					existentially choose $\event \in
  					\seteventsfor\agent$ that is executable in
  					the current pointed epistemic model
  					(fail if no such action exists);

  					universally choose a $\epistemicrelation\agent$-successor of the pointed world in the current epistemic model and make it as the new pointed world;
  				\end{algobloc}
  				
  				\algoelse \algoif $\agent \in \agtsetforall$ \algothen
  				
  				\begin{algobloc}
  					
  				universally choose $\event \in
  				\seteventsfor\agent$ that is executable in
  				the current pointed epistemic model
  				(fail if no such action exists);
  					\end{algobloc}
  				set $\kripkemodel \otimes \eventmodel, (\world, \event)$ as the current pointed epistemic model,  where $\kripkemodel, \world$ was the previous current pointed epistemic model;
  			\end{algoblocwhile}

  			\algoaccept
  			
  		\end{algoblocprocedure}
  	\end{algo}
  	\caption{Alternating algorithm for deciding in polynomial-space the distributed strategy synthesis problem when actions are public actions.  	\label{figure:algodistrEXPTIMEpublicactions}}
  \end{figure}
  Hardness follows from Theorem~  \ref{theorem:epistemicgameproblem-publicactions-EXPTIMEcomplete}.
\end{proof}}

We now turn to a decidable case for propositional actions.

\subsection{Propositional actions+hierarchical information}
\label{section:singleplayerpropositionalevents}
% ----------------------------------------------------

We consider propositional action models, which may make the size of  epistemic
models  grow unboundedly, but where the information of the different players is
hierarchical, making it easier to synchronise the existential players' strategies.
%easier to synchronize.

According to Theorem~\ref{theorem:uniformstrategyundecidable}, the
distributed strategy synthesis problem is undecidable for propositional
actions and a two-player team $\agtsetexists=\{a,b\}$ against team $\agtsetforall=\{\forall\}$. Observe that in the proof
of Theorem~\ref{theorem:uniformstrategyundecidable}, the information of
players $a$ and $b$ is incomparable: in each round $a$ only
learns the first bit produced by $\forall$'s move, while $b$ only
learns the second bit. This cannot be simulated in games with so-called
\emph{hierarchical information}, a classic restriction to regain
decidability in multi-player games of imperfect information~\cite{peterson2002decision,PR90}.

\newcommand{\subseteqwithsomespace}{~\!\!\!\subseteq~\!\!\!}

We say that an input  of the
distributed strategy synthesis problem $((\epsmodel,\world_\init),\eventmodel,\phi)$ presents 
\emph{hierarchical information} if  the set of
$\agtsetexists$ can be totally ordered ($\agent_1 < \ldots <\agent_n$) so that  $\epistemicrelationequiv[\agent_i]\subseteqwithsomespace
\epistemicrelationequiv[\agent_{i+1}]$ and $\epistemicrelationeventsequiv[\agent_i]\subseteqwithsomespace
\epistemicrelationeventsequiv[\agent_{i+1}]$, for each $1\le i<n$.

\begin{theorem}
	\label{theorem:uniformstrategydecidable}
	Distributed strategy synthesis with propositional actions and hierarchical
        information is decidable. % in $\kEXPTIME[(k+1)]$ for objectives of modal depth  $k$.
\end{theorem}

% \begin{proof}
% According to  Proposition~\ref{prop-regular}, if $\eventmodel$ is
% propositional, one can construct a finite representation of $\ETLforest{\epsmodel}{\eventmodel}$ in the form
% of an epistemic temporal model $\etlmodel$. The valuations of
% propositions $\variableturn_\agent$ allow us to see this model as a
% turn-based game structure. We conclude by observing that  turn-based games are a particular
% case of concurrent games, and recalling that distributed strategy synthesis for
% epistemic temporal objectives is decidable on concurrent game
% structures when information among the existential players is
% hierarchical~\cite{DBLP:conf/kr/MaubertM18}.
% \end{proof}

We end the section by sketching the proof of
Theorem~\ref{theorem:uniformstrategydecidable}. We
start by introducing a multi-player variant of the epistemic game arenas from
Definition~\ref{def-etl-model}. %, in which the players are the agents.

\begin{definition}
  \label{def-multi-etl-model}
  A \emph{multi-player epistemic  game arena} 
  $\etlmodel=(\setworlds,\world_\init,\trans,(\epistemicrelationequiv)_{\agent
    \in \agtset},\valworlds)$ is a structure such that
  \begin{itemize}
  \item   $(\setworlds,(\epistemicrelationequiv)_{\agent
    \in \agtset},\valworlds)$ is an epistemic model, 
\item   $\setworlds=\uplus_{\agent\in\agtset}\setworlds_\agent$ is
  partitioned into  positions of each agent,
\item  $\world_\init$ is
  an \emph{initial world} and
\item  $\trans\subseteq\setworlds\times\setworlds$ is a
  \emph{transition relation}.
  \end{itemize}
\end{definition}

Accessibility relations $\epistemicrelationequiv$ are extended to
histories, strategies of agent $a$ are required to be uniform with
respect to $\epistemicrelationequiv$, and the notions of outcomes,
distributed strategies and winning distributed strategies are defined as before.

Theorem~\ref{theorem:uniformstrategydecidable} is established by
reducing the distributed strategy synthesis problem to a similar
problem in multi-player epistemic games, known to be decidable:
% with hierarchical information and epistemic temporal objectives,
% recently proved decidable: 

\begin{theorem}[\cite{DBLP:conf/mfcs/Puchala10,DBLP:conf/kr/MaubertM18}]
  \label{theo-unif-strat-KR}
Existence of winning distributed strategies in  multi-player epistemic
  games % $\etlmodel$
  with hierarchical information and epistemic temporal objectives % of
  % modal depth $k$
  is
  decidable.
\end{theorem}

% \islongversion{See also~\cite{DBLP:conf/mfcs/Puchala10} for a
% weaker variant of this result.}

The reduction is very similar to the one in the proof of Proposition~\ref{prop-regular}.
The main difference is that we use variable $\variableturn$ instead of bit $i\in\{0,1\}$ to define the
positions of the different agents.  The imperfect information of
players is defined based on the accessibility relations in
$\kripkemodel$ and $\eventmodel$.
More precisely, we can show that:

\begin{proposition} % Proposition 27, sec 7.2.3
  \label{prop-regular-imp}
  Given an instance $((\epsmodel,\world_\init),\eventmodel,\phi)$ of
  the distributed strategy synthesis problem 
  where $\eventmodel$ is propositional, one can construct a multi-player epistemic  game arena $\etlmodel$
  such that
  % $|\etlmodel|\le
% |\epsmodel|+|\eventmodel|2^{|\AP|}$ and 
the distributed strategy synthesis problem 
 is equivalent to the existence of a winning distributed strategy for
 $\agtsetexists$ to enforce $\phi$
in $\etlmodel$ and 
$|\etlmodel|\le |\epsmodel|+|\eventmodel|\times 2^{m}$, where $m$ is
the number of atomic propositions involved.
\end{proposition}

\begin{proof}
  Let $\epsmodel=(\setworlds,(\epistemicrelationequiv)_{\agent
    \in \agtset},\valworlds)$ be an S5 epistemic
  model, 
  $\eventmodel=(\setevents,\{\epistemicrelationeventsequiv\}_{\agent\in\agents},\pre,\post)$
  an S5 propositional
  event model with
  $\setevents=\uplus_{\agent\in\agtset}\setevents_\agent$, and
  $\world_\init\in\setworlds$ an initial world.
Let $\atmsetf$ be the atomic propositions involved.

We define the multiplayer epistemic game arena $\etlmodel=(\setworlds',\world_\init,\trans,(\epistemicrelationequiv')_{\agent
    \in \agtset},\valworlds')$ as follows. First,
let
 $\setworlds'=\setworlds\union
    \setevents\times 2^\AP$, and for each $\agent\in\agtset$
let    $\setworlds'_\agent=\set{\world \mid \world\models
      \is\variableturn\agent}\cup \set{(\event,\val)\mid \val \models \is\variableturn\agent}$.
Next, for transitions, we let $(\world,(\event,\val))\in \trans$ if     
 $\world \models \pre(\event)$ 
 and $\val=\postval(\world,\event)$, and 
$((\event,\val),(\event',\val'))\in\trans$ if 
$\val \models \pre(\event')$ and
$\val'=\postval(\val,\event)$.
Next, for each $\agent\in\agents$, we let
 $\epistemicrelationequiv'\;=\;\epistemicrelationequiv{}\union
  \{((\event,\val),(\event',\val'))\mid
  (\event,\event')\in\releventsi\}$, and finally
 $\valworlds'(\world)=\valworlds(\world)$ and $\valworlds'(\event,\val)=\val$.

One can see that the structure given by the set of histories
$\Hist{\etlmodel}$ and relations $\epistemicrelationequiv'$
extended to these histories is isomorphic to
$\ETLforest{\epsmodel}{\eventmodel}$, and histories that satisfy
$\is\variableturn\agent$ in $\ETLforest{\epsmodel}{\eventmodel}$ correspond to histories
that also satisfy
$\is\variableturn\agent$ in $\etlmodel$
(provided they start in $\world_\init$), and we are done.
\end{proof} %}

%%% Local Variables:
%%% mode: latex
%%% TeX-master: "main"
%%% End:

%\section{Discussions}
%%======================
%\label{sec-discussion}
%\bm{dire que quand il existe des stratégies on peut en synthétiser des
%régulières}
%
%\fs{comparaison avec Reif à faire \cite{DBLP:journals/jcss/Reif84}}
%
%\fs{discussion sur la concision entre les deux framework "automatic games" et "DEL games"}
%
%\fs{pour la décidabilité, avec DEL, faut choisir : 2 agents et propositionnel ou alors carrément borner le temps}
%
%\fs{c'est décidabilité les jeux si on autorise des actions plus faibles comme  annonces publiques de Tiago De Lima (c'est un cadre simple pour avoir décidabilité) \cite{DBLP:journals/logcom/Lima14}}
%

\section{Perspectives}
%======================
\label{section:perspectives}

%We plan to implement a generic reasoner..........
%\fs{implémentation}
%\fs{citer Vardi et symbolic automata}
%\fs{citer modele symbolique de DEL}
%
We have incrementally extended the framework of epistemic planning to a game setting
where players’ actions are described by action models from
DEL. We have established
fine-grained results depending on the type of action models.

Works on classical planning that consider game features
exist, and they can easily be located in the landscape of decision problems we
have considered. Typically, our controller (resp. distributed strategy) synthesis problem  subsumes
\emph{conditional planning} with full (resp. partial) observability \cite{DBLP:conf/aips/Rintanen04a}. %that corresponds to a single-world initial
%pointed epistemic model.
  Also, \emph{conformant planning} (partial
information where the plan is a sequence of actions) corresponds to a
particular case of our distributed strategy synthesis problem where
$\agtsetexists = \set{\exists}$, $\agtsetforall = \set{\forall}$ are
singletons, and $\exists$ is blind, \ie all actions in
$\seteventsfor{\forall}$ are indistinguishable for her. Blindness and
uniformity assumption  make that the strategies of $\exists$ can be seen as
sequential plans.

Moreover, the decision problems we have considered go well beyond classical
planning by addressing, \eg distributed planning with cooperative
or/and adversarial features. % as the distributed strategy synthesis
% problem, which also joins game theory issues. There is evidence with
We are thus confronted in a DEL setting to issues usually met in game
theory, 
as witnessed by the undecidability of the distributed strategy synthesis problem
 for rather simple
action models (Theorem~\ref{theorem:uniformstrategyundecidable}). % , inherited from a game theory standard.
However, the DEL
perspective we propose offers a language to specify actions
(preconditions, postconditions, and epistemic relations between
actions) that may help identifying yet unknown decidable
cases. %The quest of other restrictions over the action models that ensure decidability is open.

We note that our results should transfer to  Game Description Language,
  equivalent to DEL
\cite{DBLP:conf/ijcai/EngesserMNT18}.  %, even though we did not here formally state it.

One interesting extension of the unifying setting of DEL games
would be to consider  concurrent
games, where players execute actions concurrently, but this will require
to first generalise the product operation of DEL. Another direction would be
to consider richer objectives such as ones expressed in epistemic
temporal logic, which can express not only  reachability 
but also safety or liveness objectives for instance.

%define new types of update operations in DEL to model execution of concurrent actions.

% This equivalence has to be refined: let $GDL_{prop}$ be the fragment of GDL that corresponds to propositional DEL. Therefore, our The decidability provided in Theorem~\ref{theorem:epistemicgamedecidable} would imply that reasoning about unbounded two-player games in $GDL_{prop}$ is also decidable.

Our approach contributes to putting closer the
field of  multi-agent planning in AI with the  field of multi-player
games in formal methods.
The  setting of DEL games may be beneficial to both, allowing the
transfer of powerful automata and game techniques from formal
methods to epistemic planning, and bringing in multiplayer games new insights from the  fine modelling offered by DEL.

%%
%% Bibliography
%%

%% Please use bibtex, 

%\newpage
\bibliography{biblio}

\end{document}